\documentclass[a4paper]{article}
\usepackage[margin=0.75in]{geometry}

\usepackage{tocloft}
\usepackage{graphicx}
\usepackage{amsmath}
\usepackage{amsfonts}
\usepackage{mathrsfs}
\usepackage[all,cmtip]{xy}
\usepackage{soul}
\usepackage{epigraph}
\usepackage{hyperref}
\usepackage{amsthm}

\newcounter{theorem}
\newtheorem{thm}{Theorem}[section]
\newtheorem{corol}{Corollary}[theorem]
\newtheorem{lemma}[theorem]{Lemma}
\newtheorem{propos}[theorem]{Proposition}
\newtheorem{defin}[theorem]{Definition}
\newtheorem*{remark}{Remark}

\title{Topological Information Data Analysis}
\author{Pierre Baudot$^1,2$, Monica Tapia$^2$, Daniel Bennequin$^3$ and Jean-Marc Goaillard$^2$ \\$^1$Median technologies, Sophia antiplis, France,\\ $^2$Inserm UNIS UMR1072 - Universit\'{e} Aix-Marseille AMU, Marseille, France\\ $^3$Université Paris Diderot, Institut Mathématique de Jussieu, Paris, France.\\
	pierre.baudot@gmail.com}

\date{5th July 2019}

\begin{document}
	\maketitle
	
	\begin{abstract}
		This paper presents methods that quantify the structure of statistical interactions within a given data set, and was first used in \cite{Tapia2018}. It establishes new results on the $k$-multivariate mutual-informations ($I_k$) inspired by the topological formulation of Information introduced in \cite{Baudot2015a,Vigneaux2017}. In particular we show that the vanishing of all $I_k$  for $2\leq k \leq n$ of $n$ random variables is equivalent to their statistical independence. Pursuing the work of Hu Kuo Ting and Te Sun Han \cite{Hu1962,Han1975,Han1978}, we show that information functions provide co-ordinates  for binary variables, and that they are analytically independent on the probability simplex for any set of finite variables.  The maximal positive $I_k$ identifies the variables that co-vary the most in the population, whereas the minimal negative $I_k$ identifies synergistic clusters and the variables that differentiate-segregate the most the population. Finite data size effects and estimation biases severely constrain the effective computation of the information topology on data, and we provide simple statistical tests for the undersampling bias and the k-dependences following \cite{Pethel2014}. We give an example of application of these methods to genetic expression and unsupervised cell-type classification. The methods unravel biologically relevant subtypes, with a sample size of $41$ genes and with few errors. It establishes generic basic methods to quantify the epigenetic information storage and a unified epigenetic unsupervised learning formalism. We propose that higher-order statistical interactions and non identically distributed variables are constitutive characteristics of biological systems that should be estimated in order to unravel their significant statistical structure and diversity. The topological information data analysis presented here allows to precisely estimate this higher-order structure characteristic of biological systems.
	\end{abstract}

\epigraph{\textit{"When you use the word information, you should rather use the word form"}}{Ren\'{e} Thom}

	\tableofcontents

\section{Introduction}

This note presents a method of statistical analysis of a set of collected characters in a population, describing a kind of topology of the distribution of information in the data. New theoretical results are developed to justify the method. The data that concern us are represented by certain (observed or computed) parameters $s_1,...,s_n$ belonging to certain finite sets $E_1,...,E_n$ of respective cardinalities $N_1,...,N_n$, which depend on an element $z$ of a certain set $Z$, representing the tested population, of cardinality $N_Z$. In other terms we are looking at $N$ "experimental" functions $X_i:Z\rightarrow E_i, i=1,...,n$, then we will  refer to the data by the letters $(Z,X)$, where $X$ is the product function of the $X_i$, going from $Z$ to the product $E$ of all the sets $E_i,i=1,...,n$.\\
For instance, as in $[..]$, each $E_i,i=1,...,n$ has cardinality $8$ and is identified with the subset of integers $[8]=\{1,...,8\}$, each $s_i=S_i(z),i=1,...,n$ measures the level of expression of a gene $g_i$ in a neuron $z$ belonging to a set $Z$ of classified dopaminergic neurons (DA); another system to be compared to this one is made by the analog measurements for neurons $z'$ belonging to a set $Z'$ of classified non-dopaminergic neurons (NDA). To be precise, in this example, $n=41$, $N_Z=111$, $N_{Z'}=37$.\\

\indent The most usual hypothesis for interpreting the data is the existence of an objective and immutable joint probability $\mathbb{P}_X$ on the set $E$, coming from a hypothetical set $\Omega$ enclosing $Z$ and governing the results of the sample $(Z,X)$. Then, without any other knowledge, assuming that different experimental $z$ are independently chosen, the strong law of large numbers tells that the better possible approximation of $\mathbb{P}_X(s)$ is given by the number of $z$ such that $s=X(z)$ divided by the cardinality of $Z$. However, the standard inequalities of probability theory (Bienaymé-Tchebicheff, Markov, Chernov Kolmogorov, ...) show that the confidence that we can have in this approximation depends heavily on $\mathbb{P}_X$ itself. This generates a risk of circularity. \\
\indent Several approaches can be followed to escape circularity; part of them, maintaining the frequentist point of view, uses the Fisher information metric (cf. \cite{Ly2017}), part of them, using a Bayesian approach, puts probability laws on the set of probabilities themselves, and studies the evolution of these choices with the introduction of new data. Once the total joint probability is found, it is theoretically possible to verify its agreement with the marginal laws on $X_i,i=1,...,n$, or any partial joint variable $X_I=(X_{i_1},...,X_{i_k})$ for $I=\{i_1,...,i_k\}\subset [n]$, but practically it is not so easy, being one of the principal problems in Statistical Mechanics or in Bayesian Analysis.\\

\indent We will follow here a different approach, which consists in describing the manner the variables $X_i,i=1,...,n$ distribute the Information on $(Z,X)$. The experimented population $Z$ has its own characteristics that the data explore, and the frequency of every value $s_I$ of each one of the
variables $X_I,I\subset [n]$ is an information important by itself, without considering the hypothetical law on the whole set $E$. The \emph{information quantities}, derived from the Shannon entropy, offer a natural way for describing all these frequencies. In fact they define the form of the distribution of information contained in the raw data. For instance, the individual entropies $H(X_i),i=1,...,n$ tell us the shape of the individual variables: if $H(X_i)$ is small (with respect to its capacity $\log_2 N_i$), then $X_i$ corresponds to a well defined characteristic of $Z$; to the contrary if $H(X_j)$ is close to the capacity, that is the value of the entropy of the uniform distribution, the function $X_j$ corresponds to a non-trivial partition of $Z$, and does not correspond to a well defined invariant.
At the second degree, we can consider the entropies $H(X_i,X_j)$ for every pair $(i,j)$, giving the same kind of structures as before, but for pairs of variables.
To get a better description of this second degree with respect to the first one, we can look at the \emph{mutual information} as defined by Shannon,
$I(X_i;X_j)=H(X_i)+H(X_j)-H(X_i,X_j)$. If it is small, near zero, the variables are not far from being independent, if it is maximal, i.e. not far from
the minimum of $H(X_i)$ and $H(X_j)$, this means that one of the variables is almost determined by the other.
In fact $I(X_i;X_j)$ can be taken as a measure of dependence, due to its universality and its invariance. Consider the graph with vertices $X_i,i=1,...,N$ and edges $(X_i,X_j), i\neq j, i,j=1,...,N$; by labeling each vertex with the entropy and each edge with the mutual information, we get
a sort of one-dimensional skeleton of the data $(Z,X)$.
The information of higher degrees define in an analogous manner the higher dimensional skeletons of the data $(Z,X)$ (see figure \ref{figure_Supp_Result_dopa_nondopa_infopath} for example).\\

Works of Clausius, Boltzmann, Gibbs and Helmholtz underlined the importance of entropy and free energy in Statistical Physics. In particular, Gibbs gave the general definition of the entropy for the distribution of microstates, cf. \cite{Gibbs1902}. Later Shannon recognized in this entropy the basis of Information theory in his celebrated work on the mathematical theory of communication \cite{Shannon1948} (equation \ref{higher entropy}), and then further developed their structure in the lattice of variables \cite{Shannon1953}. Note that this kind of lattice takes its roots in the work of Boole on Logic and Probability \cite{Boole1854}. Defining the communication channel, information transmission and its capacity, Shannon also introduced to degree two (pairwise) mutual information functions \cite{Shannon1948}.\\
\indent The expression and study of multivariate higher degree mutual-informations (equation \ref{higher information}) was achieved in  two seemingly independent works: 1) McGill (1954) \cite{McGill1954} (see also Fano (1961) \cite{Fano1961}) with a statistical approach, who called these functions "interaction information", and 2) Hu Kuo Ting (1962) \cite{Hu1962} with an algebraic approach who also first proved the possible negativity of mutual-informations for degrees higher than $2$. The study of these functions was then pursued by Te Sun Han \cite{Han1975,Han1978}. \\
\indent Higher-order mutual-informations were then rediscovered in several different contexts, notably by Matsuda in 2001 in the context of spin glasses, who showed that negativity is the signature of frustrated states \cite{Matsuda2001} and by Bell in the context of Neuroscience, Dependent Component Analysis and Generalised Belief Propagation on hypergraphs \cite{Bell2003}. Brenner and colleagues have observed and quantified an equivalent definition of negativity of the $3$ variables mutual information, noted $I_3$, in the spiking activity of neurons and called it synergy \cite{Brenner2000}. Anastassiou and colleagues unraveled $I_3$ negativity within gene expression, corresponding in that case to cooperativity in gene regulation \cite{Watkinson2009,Kim2010}. \\
\indent Another important family of information functions, named  "total correlation", which corresponds to the difference between the sum of the entropies and the entropy of the joint, was introduced by Watanabe in 1960 \cite{Watanabe1960}. These functions were also rediscovered several times, notably by Tononi and Edelman who called them "integrated information" \cite{Tononi1998} in the context of consciousness quantification, and by Studen\'{y} and  Vejnarova \cite{Studeny1999} who called them "multi-information" in the context of graphs and conditional independences. Closely related to what we present here with respect to both the kind of data analyzed and the conclusions, Margolin and colleagues \cite{Margolin2010} used these functions of Watanabe to quantify higher order statistical dependences within genetic expression.\\

\indent In our approach, for any data $(Z,S)$, the full picture can be represented by a collection of numerical functions on the faces of a simplex $\Delta([n])$ having vertices corresponding to the random variables $X_1,...,X_n$.	We decided to focus on two subfamilies of Information functions: the first is the collection of entropies of the joint variables, denoted $H_k,k=1,...,n$, giving the numbers $H_k(X_{i_1};...;X_{i_k})$, and the degree $k$ information of the joint variables, denoted $I_k,k=1,...,n$, and giving the numbers $I_k(X_{i_1};...;X_{i_k})$; see the following section for their definition and their elementary properties.
In particular, the value on each face of a given dimension of these functions gives interesting curves (histograms, see section on statistics \ref{The independence criterion}) for testing the departure from independence, and their means over all dimensions for testing the departure from uniformity of the variables. These functions are information co-chains of degree $k$ (in the sense of ref \cite{Baudot2015a}) and have nice probabilistic interpretations. By varying in all possible manners the ordering of the variables, i.e. by applying all the permutations $\sigma$ of $[n]=\{1,...,n\}$, we obtain $n!$ paths $H_k(\sigma)$, $I_k(\sigma)$, $k=1,...,n$. They constitute respectively the \emph{$H_k$-landscape} and the \emph{$I_k$-landscape} of the data.\\
\indent When the data correspond to uniform and independent variables, that is the uninteresting null hypothesis, each path is monotonic, the $H_k$ growing linearly and the $I_k$ being equal to zero for $k$ between $2$ and $n$. Any departure from this behavior (estimated for instance in Bayesian probability on the allowed parameters) gives a hint of the \emph{form of information} in the particular data.\\
Especially interesting are the maximal paths, where $I_k(\sigma)$ decreases, being strictly positive, or strictly negative after $k=3$. Other kinds of paths could also be interesting, for instance the paths with the maximal total variation as they can be oscillatory. In the examples provided here and in \cite{Tapia2018}, we proposed to stop the empirical exploration of the information paths to their first minima, a condition of vanishing of conditional mutual-informational (conditional independence). \\

\indent As a preliminary illustration of the potential interest of such functions for general Topological Data Analysis, we quantify the information structures for the empirical measures of the expression of several genes in two pre-identified populations of cells presented in \cite{Tapia2018}, and we consider here both cases where genes or cells are considered as variables for gene or cell unsupervised classification tasks respectively.\\
In practice, the cardinality $N_Z$ of $Z$ is rather small with respect to the number of free parameters of the possible probability laws on $E$, that is $N-1=N_1...N_n-1$, then the quantities $H_k, I_k$ for $k$ larger than a certain $k_{u}$ have in general no meaning, a phenomenon commonly called undersampling or curse of dimensionality. In the example, $n$ is $20$, but $k_{u}$ is $11$. Moreover, the permutations $\sigma$ of the variables values can be applied to test the estimation of the dependences quantified by the $I_k$ against the null hypothesis of randomly generated statistical dependences. In this approach describing the raw data for themselves, undersampling is not a serious limitation. However, it is better to test the stability of the shape of the landscapes by studying random subsets of $Z$. Moreover, the analytic properties of $H_k$ and $I_k$ considered as functions of $P$ in a given face of the simplex of probabilities $\Delta([n])$ ensure that, if $P_X$ tends to $\mathbb{P}$ in this face, the shape is preserved.\\

The originality of our method is the systematic consideration of the entropy and the information landscapes and paths that can be associated to all possible permutations of the basic variables, and the extraction of exceptional paths from them, in order to define the overall form of the distribution of information among the set of variables. This new perspective has its origin in the local (topos) homological theory introduced in \cite{Baudot2015a}. Moreover this method was successfully applied to a concrete problem of gene expression in \cite{Baudot2018,Tapia2018}.\\

In the present article, we first remind the definitions and basic properties of the entropy and information chains and functions. We give equivalent formulations of the Hu Kuo Ting theorem \cite{Hu1962}, which allows to express every partial mutual conditioned higher information of collections of joint variables from elementary higher entropies $H_k(X_I)$ or by elementary higher mutual information functions $I_k(X_I)$, i.e. the functions that form the entropy landscape and information landscape, respectively. \\
\indent Second we establish that these "pure" functions are analytically independent as functions of the probability laws, in the interior of the large simplex $\Delta([n])$. This follows from the fact we also prove here, that these functions constitute coordinates (up to a finite ambiguity) on $\Delta([n])$ in the special case of binary variables $X_i,i=1,...,n$. In addition, we demonstrate that, for every set of numbers $N_i,i=1,...,n$, the cancellation of the functions $I_k(X_I), k\geq 2, I\subset [n]=\{1,...,n\}$ is a necessary and sufficient condition of the set of variables $X_1,...,X_n$ to be statistically independent. We were not able to find these results in the literature. They generalize results of Te Sun Han \cite{Han1975,Han1978}.\\
\indent Then this article not only presents a method of analysis but it gives proofs of basic results on information quantities that, to our knowledge, were not available until now in the literature.\\
\indent Third we study the statistical properties of the entropy and information landscapes and paths, and present the computational aspects. The mentioned examples of genetic expression are developed. Finally in an appendix, we show how these functions appear in the theory of Free energies, in Statistical Physics and in Bayesian Variational Analysis.\\

\section{Results}

\subsection{Entropy and Information functions}

Given a probability law $P_X$ on a finite set $E=E_X$, Shannon defined the information content of this law by the Boltzmann-Gibbs entropy \cite{Shannon1948}: 
\begin{equation}
H(P_X)=- \sum_{x\in E} P_X(x)\log_2 P_X(x).
\end{equation}
Shannon himself gave an axiomatic justification of this choice, that was developed further by Khinchin, Kendall and other mathematicians, see \cite{Khinchin1957}.\\
The article \cite{Baudot2015a} (Baudot and Bennequin) presented such a set of axioms inspired by algebraic topology, see also \cite{Vigneaux2017} (Juan-Pablo Vigneaux). In all these approaches, the
fundamental ingredient is the decomposition of the entropy for the joint variable of two variables. To better formulate this decomposition, we have proposed to consider the entropy as a function
of three variables: first a finite set $E_X$, second a probability law $P$ on $E_X$ and third, a random variable on $E_X$, i.e. a surjective map $Y:E_X \rightarrow E_Y$, considered only through
the partition of $E_X$ that it gives, indexed by the elements $y$ of $E_Y$. In this case we say that $Y$ is less fine than $X$, and write $Y\leq X$, or $X\rightarrow Y$. Then we define the entropy of $Y$ for $P$ at $X$:
\begin{equation}
H_X(Y;P)=H(Y_*(P)),
\end{equation}
where $Y_*(P)$ is the image law, also named the \emph{marginal} of $P$ by $Y$:
\begin{equation}
Y_*(P)(y)=\sum_{x|Y(x)=y}P(x).
\end{equation}
\begin{remark} 
Frequently, when the context is clear, we simply write $H_X(Y;P)=H(Y;P)$ or even $H(Y)$, as everybody does, however the "homological nature" of $H$ can only be understood with the index $X$, because it is here that the topos theory appears, see \cite{Baudot2015a,Vigneaux2017}.
\end{remark}
The second fundamental operation on probabilities (after marginalization) is the \emph{conditioning}: given $y\in E_Y$, such that $Y_*(P)(y)\neq 0$, the conditional probability $P|(Y=y)$
on $E_X$ is defined by the following rules:
\begin{align*}
\forall x| Y(x)=y,\quad &P|(Y=y)(x)=P(x)/Y_*(P)(y),\\
\forall x| Y(x)\neq y,\quad &P|(Y=y)(x)=0
\end{align*}
This allows to define the conditional entropy, as Shannon has done, for any $Z$ and $Y$ both less fine that $X$,
\begin{equation}
Y.H(Z;P)=\sum_{y\in E_Y} H(Z;P|(Y=y))Y_*(P)(y).
\end{equation}
Note that if $P|(Y=y)$ is not well defined, we can simply forget it in the formula, because it appears multiplied by zero.\\
This operation is associative (see \cite{Baudot2015a,Vigneaux2017}), i.e. for any triple $W,Y,Z$ of variables less fine than $X$,
\begin{equation}
(W,Y).H(Z;P)=W.(Y.H)(Z;P).
\end{equation}
With these notations, the fundamental functional equation of Information Theory, or its first axiom, according to Shannon, is
\begin{equation}\label{fundeqinfo}
H((Y,Z);P)=H(Y;P)+Y.H(Z;P).
\end{equation}
\begin{remark} 
In \cite{Baudot2015a,Vigneaux2017} it is shown that this equation can be understood as a co-cycle equation of degree one of a module in a topos, in the sense of Grothendieck and Verdier \cite{Artin1964}, and why the entropy is generically the only universal generator of the first co-homology functor.
\end{remark}

More generally, we consider a collection of sets $E_X,X\in\mathcal{C}$, such that each time $Y,Z$ are less fine than $X$ and belong to $\mathcal{C}$, then $(Y,Z)$ also belongs to $\mathcal{C}$; in this case we name $\mathcal{C}$ an \emph{information category}. An example is given by the joint variables $X=(X_{i_1},..., X_{i_m})$ of $n$ basic variables $X_1,...,X_n$ with values in finite sets $E_1,...,E_n$; the set $E_X$ being the product $E_{i_1}\times...\times E_{i_m}$.\\
Then for every natural integer $k\geq 1$, we can consider families indexed by $X$ of (measurable) functions of the probability $P_X$ that are indexed by several variables $Y_1,...,Y_k$
less fine than $X$
\begin{equation}	
P_X\mapsto F_X(Y_1;...;Y_k;P_X)
\end{equation}
satisfying the compatibility equations;
\begin{equation}
\forall X',X\leq X', \forall P_{X'},\quad  F_X(Y_1;...;Y_k;X_*(P_{X'}))=F_{X'}(Y_1;...;Y_k;P_{X'}).
\end{equation}
We call these functions the \emph{co-chains} of degree $k$ of $\emph{C}$ for the probability laws. An equivalent axiom is that $F_X(Y_1;...;Y_k;P_X)$ only depends on the image of $P_X$ by the joint variable $(Y_1,...,Y_k)$. We call this property \emph{locality} of the family $F=(F_X, X\in\mathcal{C})$.\\

The action by conditioning extends verbally to the co-chains of any degree:\\
if $Y$  is less fine than $X$,
\begin{equation}
Y.F_X(Y_1;...;Y_k;P)=\sum_{y\in E_Y} F_X(Y_1;...;Y_k;P|(Y=y))Y_*(P)(y).
\end{equation}
It satisfies again the associativity condition.\\
Higher mutual information quantities were defined by Hu Kuo Ting \cite{Hu1962} and McGill \cite{McGill1954}, generalizing the Shannon mutual information \cite{Baudot2015a,Tapia2018}:\\
in our terms, for $k$ random variables $X_1,...,X_k$ less fine than $X$ and one probability law $P$ on the set $E_X$,
\begin{equation}
H_k(X_1;...;X_k;P)=H((X_1,...,X_k);P).
\end{equation}
And more generally, for $j\leq k$, we define
\begin{equation} \label{higher entropy}
H_j(X_1;...;X_k;P)=\sum_{I\subset [k];card(I)=j}H(X_I;P),
\end{equation}
where $X_I$ denotes the joint variable of the $X_i$ such that $i\in I$.\\
We name these functions of $P$ the \emph{joint entropies}. \\
Then the higher information functions are defined by
\begin{equation}\label{higher information}
I_k(X_1;...;X_k;P)=\sum_{j=1}^{j=k}(-1)^{j-1}H_j(X_1;...;X_k;P),
\end{equation}
In particular we have $I_1=H$, the usual entropy.\\
\noindent Reciprocally the functions $I_k$ decompose the entropy of the finest joint partition:
\begin{equation}
H(X_1;X_2;...;X_n;\mathbb{P})= \sum_{k=1}^{k=n}(-1)^{k-1}\sum_{I\subset [n];card(I)=k}
I_{k}(X_{i_1};X_{i_2};...;X_{i_k};\mathbb{P})
\end{equation}

\noindent The following result is immediate from the definitions, and the fact that $H_X,X\in \mathcal{C}$ is local:\\
\begin{propos}
The joint entropies $H_k$ and the higher information quantities $I_k$ are information co-chains, i.e. they are local functions of $P$.\\
\end{propos}
\begin{remark} 
From the computational point of view, locality is important, because it means that only the less fine marginal probability has to be taken into account.
\end{remark}

The definition of $H_j,j\leq k$ and $I_k$ makes evident that they are symmetric functions, i.e. they are invariant by every permutation of the letters $X_1,...,X_k$.\\
The particular case $I_2(S;T)=H(S)+H(T)-H(S,T)$ is the usual mutual information defined by Shannon.\\
Using the concavity of the logarithm, it is easy to show that $I_1$ and $I_2$ have only positive values, but this ceases to be true for $I_k$ as soon as $k\geq 3$ \cite{Hu1962,Matsuda2001}.\\

Hu kuo Ting defined in \cite{Hu1962} other information quantities, by the following formulas:
\begin{equation}
I_{k,l}(Y_1;...;Y_k;P_X|Z_1,...,Z_l)=(Z_1,...,Z_l).I_k(Y_1;...;Y_k;P_X).
\end{equation}
For instance, considering a family of basic variables $X_i,i=1,...,n$,
\begin{equation}
I_{k,l}(X_{I_1};...;X_{I_k};(\mathbb{P}|X_J))=X_J.I_k(X_{I_1};...;X_{I_k};\mathbb{P}),
\end{equation}
for the joint variables $X_{I_1},...,X_{I_k}, X_J$, where  $I_1,...,I_k,J \subset [n]$.\\

\noindent The following remarkable result is due to Hu Kuo Ting \cite{Hu1962}:\\
\begin{thm} \label{thm Hu}
Let $X_1,...,X_n$ be any set of random variables and $\mathbb{P}$ a given probability on the product $E_X$ of the respective images $E_1,...,E_n$, then there exist finite sets $\Sigma_1,...,\Sigma_n$ and a numerical function $\varphi$ from the union $\Sigma$ of these sets to $\mathbb{R}$, such that for any collection of subsets $I_m;m=1,...,k$ of $\{1,...,n\}$, and any subset $J$ of $\{1,...,n\}$ of cardinality $l$, the following identity holds true
\begin{equation}
I_{k,l}(X_{I_1};...;X_{I_k};(\mathbb{P}|X_J))=\varphi(\Sigma_{I_1}\cap...\cap
\Sigma_{I_k}\backslash \Sigma_J),
\end{equation}
where we have denoted $X_I=(X_{i_1},...,X_{i_l})$ and $\Sigma_I=\Sigma_{i_1}\cup...\cup \Sigma_{i_l}$ for $I=\{i_1,...,i_l\}$, and where $\Omega\backslash \Sigma_J$ denotes the set of points in $\Omega$ that do not belong to $\Sigma_J$, i.e. the set $\Omega\cap (\Sigma\backslash \Sigma_J)$, named subtraction of $Y=\Sigma_J$ from $\Omega$.
\end{thm}

\noindent The Hu Kuo Ting theorem says that for a given joint probability law $\mathbb{P}$, and from the point of view of the information quantities $I_{k,l}$, the joint operation of variables corresponds to the union of sets, the graduation $k$ corresponds to the intersection, and the conditioning by a variable corresponds to the difference of sets. This can be precisely formulated as follows:

\begin{corol}\label{corol1}
Let $X_1,...,X_n$ be any set of random variables on the product $E_X$ of the respective goals $E_1,...,E_n$, then for any probability $\mathbb{P}$ on $E_X$, every universal identity between disjoint sums of subsets of a finite set that are obtained, starting with $n$ subsets $\Sigma_1,...,\Sigma_n$, by 1) forming collections of reunions, 2) taking successive intersections of these
unions, and 3) subtracting by one of them, gives an identity between sums of information quantities, by replacing the union by the joint variables $(.,.)$, the intersections by the juxtaposition $(.;.;.)$
and the subtraction by the conditioning.
\end{corol}

\begin{remark} 
Conversely the corollary implies the Theorem.
\end{remark}

This corollary is the source of many identities between the information quantities.\\
For instance, the fundamental equation \eqref{fundeqinfo} corresponds to the fact that
the union of two sets $A,B$ is the disjoint union of one of them, say $A$ and of the difference of the
other with this one, say $B\backslash A$.\\
\noindent The following formula follows from \eqref{fundeqinfo}:
\begin{multline}\label{fundeqinfomult}
H_{k+1}(X_0;X_1;...;X_k;\mathbb{P})=H_k((X_0,X_1);X_2;...;X_k;\mathbb{P})\\=H_k(X_1;...;X_k;\mathbb{P})+X_0.H_k(X_1;...;X_k;\mathbb{P}).
\end{multline}

\noindent The two following identities are also easy consequences of the Corollary $1$; they are important for the method of data analysis presented in this article:\\

\begin{propos} 
Let $k$ be any integer
\begin{multline}\label{nonlocfundeqinfo}
I_k((X_0,X_1);X_2;...;X_k;\mathbb{P})=I_k(X_0;X_2;...;X_k;\mathbb{P})+X_0.I_k(X_1;X_2;...;X_k;\mathbb{P})
\end{multline}
\end{propos}

\begin{propos} 
Let $k$ be any integer
\begin{equation}\label{recinfok}
I_{k+1}(X_0;X_1;...;X_k;\mathbb{P})=I_k(X_1;X_2;...;X_k;\mathbb{P})-X_0.I_k(X_1;X_2;...;X_k;\mathbb{P}).
\end{equation}
\end{propos}

\begin{remark} Be careful that some universal formulas between sets do not give identities between information functions; for instance $A\cap (B\cup C)=(A\cap B)\cup (A\cap C)$ but in general we have
\begin{equation}\label{falseformula}
I_2(X;(Y,Z))\neq I_2(X;Y)+I_2(X;Z),
\end{equation}
What is true is the following identity:
\begin{equation}\label{topococyle}
I_2(X;(Y,Z))+I_2(Y;Z)=I_2((X,Y);Z)+I_2(X;Y),
\end{equation}
This corresponds to the following universal formula between sets
\begin{equation}
(A\cap (B\cup C))\cup (B\cap C)=(A\cap B)\cup ((A\cup B)\cap C).
\end{equation}
The formula \eqref{topococyle} follows directly from the definition of $I_2$, by developing the four terms of the equation. It expresses the fact that $I_2$ is a simplicial co-cycle, being the simplicial co-boundary of $H$ itself.\\
However, although this formula between sets is true, it is not of the form authorized by the corollary $1$.\\
Consequently, some identities of sets that are not contained in the Theorem \ref{thm Hu} correspond to information identities, but, as we saw just before with the false formula \eqref{falseformula}, not all identities of sets correspond to information identities.
\end{remark}

As we already said, the set of joint variables $X_I$, for all the subsets $I$ of $[n]=\{1,...,n\}$, is an information category, the set $\mathcal{C}$ being the $n-1$-simplex $\Delta([n])$ of vertices $X_1,...,X_n$. In what follows we do not consider more general information categories.\\
\indent We can paraphrase the Theorem \ref{thm Hu}, by a combinatorial Theorem on the simplex $\Delta([n])$:
\begin{defin}
Let $X_1,...,X_n$ be a set of random variables with respective goals $E_1,...,E_n$, and let $X_I=\{X_{i_1},...,X_{i_k}\}$ be a face of $\Delta([n])$, we define, for a probability $P$ on the product $E$ of all the $E_i,i=1,...,n$,
\begin{equation}
\eta_I(P)=\eta(X_{i_1};...;X_{i_k};P)=X_{[n]\backslash I}.I_k(X_{i_1};...;X_{i_k};P),
\end{equation}
\end{defin}
\begin{remark} 
With the exception $J=[n]$, the function $\eta_J$ is not an information co-chain of degree $k$. But it is useful in the demonstrations of some of the following results.
\end{remark}

Embed $\Delta([n])$ in the hyperplane $x_1+...+x_n=1$ as the standard simplex in $\mathbb{R}^{n}$ (the intersection of the above hyperplane with the positive cone, where $\forall i=1,...,n, x_i\geq 0$), and consider the balls $\Sigma_1,...,\Sigma_n$ of radius $R$ strictly larger than $\sqrt{(n-1)/n}$ that are centered on the vertices $X_j;j=1,...,n$; they have all possible non-empty intersections convex. The subsets
$\Sigma'_I=\Sigma_I\setminus  \Sigma_{[n]\backslash I}$ are the connected components of complementary set of the unions of the boundary spheres $\partial\Sigma_1,...,\partial\Sigma_n$ in the total
union $\Sigma$ of the balls $\Sigma_1,...,\Sigma_n$.\\

\begin{propos}\label{info_sums}
For every $k+1$ subsets $I_1,..,I_k,K$ of $[n]$, if $l$ denotes the cardinality of $K$, the information function $I_{k,l}(X_{I_1};...;X_{I_k};P|X_K)$ is equal to the sum of the functions $\eta_J(P)$, where $J$ describes all the faces such that $\Sigma'_J$ is one of the connected components of the set $(\Sigma_{I_1}\cap...\cap \Sigma_{I_k})\backslash \Sigma_K$.\\
\end{propos}

\begin{proof}  
Every subset that is obtained from the $\Sigma_J;J\subset [n]$ by union, intersection and difference,
repeated indefinitely (i.e. every element of the Boolean algebra generated by the $\Sigma_i;i=1,...,n$), is a disjoint union of some of the sets $\Sigma'_J$. This is true in particular for the sets obtained by the succession of operations $1,2,3$ in the order prescribed by the corollary \ref{corol1} above. Then the proposition follows from the corollary \ref{corol1}.
\end{proof}

\indent We define the \emph{elementary} (or \emph{pure}) joint entropies $H_k(X_I)$ and the elementary (or pure) higher information functions $I_k(X_I)$ as $H_k(X_{i_1};...;X_{i_k};P)$ and $I_k(X_{i_1};...;X_{i_k};P)$ respectively, where $I=\{i_1,...,i_k\}\subset[n]$ describes
the subsets of $[n]$. In the following of the text, we will consider only these pure quantities. We will frequently denote them simply by $H_k$ (resp. $I_k$).
The other information quantities use joint variables and conditioning, but the preceding result tells that they can be computed from the pure quantities. \\
For the pure functions, the decompositions in the basis $\eta_I$ are simple:\\

\begin{propos} 
If $I=\{i_1,...,i_k\}$, we have
\begin{equation}
H_k(X_{i_1};...;X_{i_k};P)=\sum_{J\subset [n]|\exists m, i_m\in J}\eta_J(P),
\end{equation}
and
\begin{equation}
I_k(X_{i_1};...;X_{i_k};P)=\sum_{J\supset I}\eta_J(P).
\end{equation}
\end{propos}
In other terms, the function $H_k$ evaluated on a face $X_I$ of dimension $k$ is given by the sum of the functions $\eta_J$ over all the faces $X_J$
connected to $X_I$. And the function $I_k$ evaluated on $X_I$ is the sum of the functions $\eta_J$ over all the faces $X_J$ that contain $X_I$.\\

\begin{propos} 
For any face $J$ of $\Delta([n])$, of dimension $l$, and any probability $P$ on $E_X$, we have
\begin{equation}
\eta_J(P)=\sum_{k\geq l}\sum_{I\supseteq J| dimI=k}(-1)^{k-l}H_k(X_I;P).
\end{equation}
\end{propos}

\begin{proof}  
This follows from the Moebius inversion formula (Rota 1964 \cite{Rota1964a}).
\end{proof}

\begin{corol} 
(Te Sun Han): Any Shannon information quantity is a linear combination of the pure functions $I_k, k\geq1$ (resp. $H_kk\geq1$), with coefficients in $\mathbb{Z}$, the ring of relative integers.
\end{corol}

\begin{proof} 
This follows from the proposition \ref{info_sums}.
\end{proof} 

\noindent Hu Kuo Ting \cite{Hu1962} also proved a remarkable property of the information functions associated to a Markov process:\\
\begin{propos} 
The variables $X_1,...,X_n$ can be arranged in a Markov process $(X_{i_1},...,X_{i_n})$
if and only if, for every subset $J=\{j_1,...,j_{k-2}\}$ of $\{ i_2,...,i_{n-1}\}$ of cardinality $k-2$, we have
\begin{equation}
I_k(X_{i_1};X_{j_1},...;X_{j_{k-2}};X_{i_n})=I_2(X_{i_1};X_{i_n}).
\end{equation}
\end{propos} 
\noindent This implies that, for a Markov process between $(X_{i_1},...,X_{i_n})$, all the functions $I_k(X_I)$ involving $i_1$ and $i_n$, are positive.

\subsection{The independence criterion}\label{The independence criterion}

The total correlations were defined by Watanabe as the difference of the sum of entropies and the joint entropy, noted  $G_k$  \cite{Watanabe1960} (see also \cite{Tononi1998,Studeny1999,Margolin2010}):
\begin{equation}\label{total correlation}
G_k(X_1;...;X_k;P)=\sum_{i=1}^k H(X_i) - H(X_1;...;X_k).
\end{equation}
Total correlations are Kullback-Leibler divergences, cf. appendix \ref{Appendix: Bayes free energy} on bayes free energy; and $I_2=G_2$.
It is well known (cf. the above references or \cite{Cover1991}) that for $n\geq 2$, the variables $X_1,...,X_n$ are statistically independent for the probability $P$, if and only if $G_n(X_1;...;X_n)=0$, i.e.
\begin{equation}\label{independententropy}
	H(X_1,...,X_n;P)=H(X_1;P)+...+H(X_n;P).
\end{equation}

\begin{remark}
The result is proved by induction using repetitively  the case  $n=2$, which comes from the
strict concavity of the function $H(P)$ on the simplex $\Delta([n])$.
\end{remark}

\begin{thm} \label{thm_indepenedence}
For every $n$ and every set $E_1,...,E_n$ of respective cardinalities $N_1,...,N_n$, the
probability $P$ renders the $n$ variables $X_i,i=1,...,n$ statistically independent if and only if the $2^{n}-n-1$ quantities $I_k$ for $k \geq 2$ are equal to zero.
\end{thm} 

\begin{proof} 
For $n=2$ this results immediately from the above criterion and the definition of $I_2$.
Then we proceed by recurrence on $n$, and assuming that the result is true for $n-1$ we deduce it for $n$.\\
The definition of $I_n$ is
\begin{multline}\label{infoN}
I_n(X_1;...;X_n;P)=H(X_1;P)+...+H(X_n;P)\\-H(X_1,X_2;P)-...+(-1)^{n+1}H(X_1,...,X_n;P).
\end{multline}
By recurrence, the quantities $I_k$ for $2\leq k\leq n-1$ are all equal to zero if and only if, for
every subset $I=\{ i_1,...,i_k\} \subset [n]$ of cardinality $k$ between $2$ and $n-1$, the variables $X_{i_1},...,X_{i_k}$ are independent.
Suppose this is the case.
In the above formula \eqref{infoN}, we can replace all the intermediary higher entropies $H(X_I;P)$ for $I$ between $2$ and $n-1$ by the corresponding
sum of the individual entropies $H(X_{i_1})+...+H(X_{i_k})$. By symmetry each term $H(X_i)$ appears the same number of time, with the same sign
each time. The total sum of signs is obtained by replacing each $H(X_i)$ by $1$; it is
\begin{equation}
\Sigma=n-2C_n^{2}+3C_n^{3}-...+(-1)^{n}(n-1)C_n^{n-1}.
\end{equation}
However, as a polynomial in $x$, we have
\begin{equation}
(1-x)^{n}=1-nx+C_n^{2}x^{2}-...+(-1)^{n}x^{n},
\end{equation}
thus
\begin{equation}
\frac{d}{dx}(1-x)^{n}=-n+2C_n^{2}x-...+(-1)^{n}nx^{n-1},
\end{equation}
therefore
\begin{multline}
n-2C_n^{2}+...+(-1)^{n}(n-1)C_n^{n-1}\\=(-1)^{n}n-\frac{d}{dx}(1-x)^{n}|_{x=1}=(-1)^{n}n,
\end{multline}
because $n\geq 2$.\\
Then we obtain
\begin{multline}
I_n(X_1;...;X_n;P)\\=(-1)^{n-1}H(X_1,...,X_n;P)+(-1)^{n}(H(X_1;P)+...+H(X_n;P)).
\end{multline}
Therefore, if the variables $X_i;i=1,...,n$ are all independent, the quantity $I_n$ is equal to $0$. And conversely,
if $I_n=0$, the variables $X_i;i=1,...,n$ are all independent.
\end{proof} 

\noindent Te Sun Han established that, for any subset $I_0$ of $[n]$ of cardinality $k_0\geq 2$, there exist probability laws such that all the $I_k(X_I), k\geq 2$ are zero with the exception of $I_{k_0}(X_{I_0})$ \cite{Han1975,Han1978}. Consequently, in the equations of the theorem \ref{thm_indepenedence}, no one can be forgotten. \\
The unique equation \eqref{independententropy} also characterizes the statistical independence, but its gradient with respect to $P$ is strongly degenerate along the variety of independent laws. As shown by Te Sun Han \cite{Han1975,Han1978}, this is not the case for the $I_k$.\\

\subsection{Information coordinates}\label{Information coordinates}

The number of different functions $\eta_I$, resp. pure $I_k$, resp. pure $H_k$, is $2^{n}-1$ in the three cases. It is natural to ask if each of these families of functions of $P_X$ are analytically independent; we will prove here that this is true. The basis of the proof is the fact that each family gives finitely ambiguous coordinates in the case of binary variables, i.e. when all the numbers $N_i,i=1,...,n$
are equal to $2$. Then we begin by considering $n$ binary variables with values $0$ or $1$. \\

\noindent Let us look first at the cases $n=1$ and $n=2$. And consider only the family $H_k$, the other families being easily deduced by linear isomorphisms.\\

\noindent In the first case the only function to consider is the entropy
\begin{multline}
H((p_0,p_1))=-p_0\log_2(p_0)-p_1\log_2(p_1)\\=-\frac{1}{\ln 2}(x\ln x-(1-x)\ln (1-x))=h(x),
\end{multline}
if we put $x=p_0$. To get a probability $(p_0,p_1)$ we impose that $x$ belongs to $[0,1]$.
As a function of $x$, $h$ is strictly concave, attaining all values between $0$ and $1$, but it is not injective, due to the symmetry $x\mapsto 1-x$, which corresponds to the exchange of the values $0$ and $1$.\\

\noindent For $n=2$, we have two variables $X_1,X_2$ and three functions $H(X_1;P)$, $H(X_2;P)$, $H(X_1,X_2;P)$. These functions are all concave and real analytic in the interior of the simplex of dimension $3$.\\
Let us describe the probability law by four positive numbers $p_{00}, p_{01}, p_{10}, p_{11}$ of sum $1$. The marginal laws for $X_1$ and $X_2$ are described respectively by the following coordinates:
\begin{align}
&p_0=p_{00}+ p_{01},\quad p_1=p_{10}+ p_{11},\\
&q_0=p_{00}+ p_{10},\quad q_1=p_{01}+ p_{11}.
\end{align}
For the values of $H(X_1;P)$ and $H(X_2;P)$ we can take independently two arbitrary real numbers between $0$ and $1$. Moreover, from the case $n=1$, if two laws $P$ and $P'$ give the same values $H_1$ and $H_2$ of $H(X_1;P)$ and $H(X_2;P)$ respectively, we can reorder $0$ and $1$ independently on each variable in such a manner that the images of $P$ and $P'$ by $X_1$ and $X_2$ coincide, i.e. we can suppose that $p_0=p'_0$ and $q_0=q'_0$, which implies $p_1=p'_1$ and $q_1=q'_1$, due to the condition of sum $1$. It is easy to show that the third function $H(X_1,X_2;P)$ can take any value between the maximum of $H_1$, $H_2$ and the sum $H_1+H_2$.\\
\begin{lemma}
There exist at most two probability laws that have the same marginal laws under $X_1$ and $X_2$ and the same value $H$ of $H(X_1,X_2)$; moreover, depending on the given values $H_1,H_2,H$ in the allowed range, both cases can happen in open sets of the simplex of dimension seven.
\end{lemma}

\begin{proof} 
When we fix the values of the marginals, all the coordinates $p_{ij}$ can be expressed linearly
in one of them, for instance $x=p_{00}$:
\begin{equation}
p_{01}=p_0-x,\quad p_{10}=q_0-x,\quad p_{11}=p_1-q_0+x.
\end{equation}
Note that $x$ belongs to the interval $I$ defined by the positivity of all the $p_{ij}$:
\begin{equation}
x\geq 0,\quad x\leq p_0,\quad x\leq q_0,\quad x\geq q_0-p_1=q_0+p_0-1.
\end{equation}
The fundamental formula gives the two following equations:
\begin{align}
H(X_1,X_2;P)-H(X_1)&=X_1.H(X_2;P)=p_0h(\frac{x}{p_0})+p_1h(\frac{q_0-x}{p_1}),\\
H(X_1,X_2;P)-H(X_2)&=X_2.H(X_1;P)=q_0h(\frac{x}{q_0})+q_1h(\frac{p_0-x}{q_1}).
\end{align}
We define the functions $f_1(x)$ and $f_2(x)$ by the two above formulas respectively.
As a function of $x$, each one is strictly concave, being a sum of strictly concave functions, thus it cannot take the same value for more than two values of $x$.\\
This proves the first sentence of the lemma; to prove the second one, it is sufficient to give examples for both situations.\\
Remark that the functions $f_1,f_2$ have the same derivative:
\begin{equation}
f'_1(x)=f'_2(x)=\log_2 (\frac{p_{01}p_{10}}{p_{00}p_{11}}).
\end{equation}
This results from the formula $h'(u)=-\log_2(u/1-u)$ of the derivative of the entropy.\\
Then the maximum of $f_1$ or $f_2$ on $[0,1]$, is attained for $p_{01}p_{10}=p_{00}p_{11}$, that is when
\begin{equation}
x(x+1-p_0-q_0)=(x-p_0)(x-q_0)\quad \Leftrightarrow \quad x=p_0q_0,
\end{equation}
which we could have written without computation, because it corresponds to the independence of the variables $X_1$, $X_2$.\\
Then the possibility of two different laws $P, P'$ in the lemma is equivalent to the condition that $p_0q_0$ belongs to the interior of $I$. This happens for instance
for $1> p_0> q_0> q_1> p_1> 0$, where $I=[q_0-p_1,q_0]$, because in this case $p_0q_0 < q_0$ and $p_1> p_1q_0$ i.e. $p_0q_0=q_0-p_1q_0> q_0-p_1$. In fact, to get
$P\neq P'$ with the same $H$, it is sufficient to take $x$ different from $p_0q_0$ but sufficiently close to it, and $H=f_2(x)+H_2$.\\
However, even in the above case, the values of $f_1$ (or $f_2$) at the extremities of $I$ do not coincide in general. Let us prove this fact. We have
\begin{multline}
f_2(q_0)=q_1h(\frac{p_0-q_0}{q_1})=q_1h(1-\frac{p_1}{q_1})=F(p_1),\\ f_2(q_0-p_1)=q_0h(1-\frac{p_1}{q_0})=G(p_1).
\end{multline}
When $p_1=0$, the interval $I$ is reduced to the point $q_0$, and $F(0)=G(0)=0$. Now fix $q_0,q_1$, and consider the derivatives of $F,G$ with respect to
$p_1$ at every value $p_1>0$:
\begin{equation}
F'(p_1)=\log_2\frac{p_0-q_0}{p_1},\quad G'(p_1)=\log_2\frac{q_0-p_1}{p_1}.
\end{equation}
Therefore $F'(p_1)< G'(p_1)$ if and only if $p_0-q_0< q_0-p_1$, i.e. $q_0> 1/2$. Then, when $q_0> 1/2$,for $p_1> 0$ near $0$, we have $F(p_1)< G(p_1)$.\\
Consequently, any value $f_2(x)$ that is a little larger than $F(p_1)$ determines a unique value of $x$. It is in the vicinity of $q_0$. Which ends the proof of the lemma.
\end{proof} 

\noindent From this lemma, we see that there exist open sets where $8$ or $4$ different laws give the same values of the three functions $H(X_1), H(X_2), H(X_1,X_2)$. In degenerate cases, we can have $4$, $2$ or $1$ laws giving the same three values.\\
\begin{thm} \label{thm_coord}
For $n$ binary variables $X_1,...,X_n$, the functions $\eta_I$, resp. pure $I_k$, resp. pure $H_k$, characterize the probability law on $E_X$ up to a finite ambiguity.\\
\end{thm}

\begin{proof} 
From the preceding section, it is sufficient to establish the theorem for the functions $H_k(X_I)$, where $k$ goes from $1$ to $n$, and $I$ describes all the subsets of cardinality $k$ in $[n]$.\\
The proof is made by recurrence on $n$. We just have established the cases $n=1$ and $n=2$.\\
For $n> 2$ we use the fundamental formula
\begin{equation}
H(X_1,...,X_n)=H(X_1,...,X_{n-1})+(X_1,...,X_{n-1}).H(X_n).
\end{equation}
By the Marginal Theorem of H.G. Kellerer \cite{Kellerer1964} (see also F. Matus \cite{Matus1988}), knowing
the $2^{n}-2$ non-trivial marginal laws of $P$, there is only one resting dimension, thus one of the coordinates $p_i$ only is free, that we denote $x$. Suppose that all the values of the $H_k$ are known, the hypothesis of recurrence tells that all the non-trivial marginal laws are known from the values of the entropy, up to a finite ambiguity. We fix a choice for these marginals. The above fundamental formula expresses $H(X_1,...,X_n)$ as a function $f(x)$ of $x$, which is a linear combination with positive coefficients
of the entropy function $h$ applied to various affine expressions of $x$; therefore $f$ is a strictly concave function of one variable, then only two values at most are possible for $x$ when the value $f(x)$ is given.
\end{proof} 

\noindent The group $\{\pm 1\}^{n}$ of order $2^{n}$ that exchanges in all possible manners the values of the binary variables $X_i,i=1,...,X_n$ gives a part of the finite ambiguity. However, even for $n=2$, the ambiguity is not associated to the action of a finite group, contrarily to what was asserted in \cite{Baudot2015a} section $1.4$. What replaces the elements of a group are partially defined operations of permutations that deserve to be better understood.\\

\begin{thm} \label{thm_grad}
The functions $\eta_I$, resp. the pure $I_k(X_I)$, resp. the pure $H_k(X_I)$, have linearly independent gradients in an open dense set of the simplex $\Delta([n])$ of probabilities on $E_X$.
\end{thm}

\begin{proof} 
Again, it  is sufficient to treat the case of the higher pure entropies.\\
We write $N=N_1...N_n$. The elements of the simplex $\Delta(N)$ are described by vectors $(p_1,...,p_N)$ of real numbers that are positive or zero, with a sum equal to $1$. The expressions $H_k(X_J)$ are real analytic functions in the interior of this simplex. The number of these functions is $2^{n}-1$.
The dimension $N-1$ of the simplex is larger (and equal only for the fully binary case), then to establish the result, we have to find a minor of size $2^{n}-1$ of the Jacobian matrix of the partial derivatives of the entropy functions with respect to the variables $p_i,i=1,...,N-1$ that is non identically zero. For any index $j$ between $1$ and $n$ choose two different values of the set $E_j$. Then apply
the theorem $2$.
\end{proof} 

\begin{remark}
\noindent This proves the fact mentioned in $1.4$ of $[3]$. 
\end{remark}

\noindent Te Sun Han established that the quantities $I_k(X_I)$ for $k\geq 2$ are functionally independent \cite{Han1975,Han1978}.\\
\begin{remark}
The formulas of $H_k(X_I)$, then of $I_k(X_i)$ and $\eta_I$, extend analytically to the open cone $\Gamma([n])$ of vectors with positive coordinates. On this cone we pose
\begin{equation}
H_0(P)=I_0(P)=\eta_0(P)=\sum_{i=1}^{n}p_i.
\end{equation}
This is the natural function to consider to account for the empty subset of $[n]$.\\
Be careful that the functions $K_k$ for $k> 0$ are no more positive in the cone $\Gamma([n])$, because the function $-x\ln x$ becomes negative for $x> 0$. In fact we have, for $\lambda\in ]0,\infty [$, and $P=(p_1,...,p_n)\in \Gamma([n])$,
\begin{equation}
H_k(\lambda P)=\lambda H_k(P)-\lambda \log_2 \lambda H_0(P).
\end{equation}
The above theorems extend to the prolonged functions to the cone, by taking into account $H_0$.
\end{remark}

\noindent Notice further properties of information quantities:\\

\noindent For $I_k$, due to the constraints on $I_2$ and $I_3$, see Matsuda \cite{Matsuda2001}, we have for any pair of variables
\begin{equation}
0\leq I_2(X_1,X_2) \leq min \{H(X_1), H(X_2)\},
\end{equation}
and any triple $X_1,X_2,X_3$:
\begin{equation}\label{infomutineq}
-min \{H(X_1), H(X_2), H(X_3)\} \leq I_3(X_1,X_2,X_3) \leq min \{H(X_1), H(X_2), H(X_3)\}.
\end{equation}
It could be that interesting inequalities also exist for $k\geq 4$, but it seems that they are unknown.\\

Contrarily to $H_k$, the behavior of the function $I_k$ is not the same for $k$ even and $k$ odd.
In particular, as functions of the probability $P_X$, the odd functions $I_{2m+1}$, for instance $I_1=H_1=H$, or $I_3$ (ordinary synergy), have properties of the type of pseudo-concave functions (in the sense of \cite{Baudot2015a}), and the even functions $I_{2m}$, like $I_2$ (usual mutual information) have properties of the type of convex functions (see \cite{Baudot2015a} for a more precise statement). Note that this accords well with the fact that the total entropy $H(X)$, which is concave, is the alternate sum of the $I_k(X_I)$ over the subsets $I$ of $[n]$,
with the sign $(-1)^{k-1}$ (cf. appendix \ref{Appendix: Bayes free energy}).\\
\indent Another difference is that each odd function $I_{2m+1}$ is an information co-cycle, in fact a co-boundary if $m\geq 1$ (in the information co-homology defined in \cite{Baudot2015a}), but each odd function $I_{2m+1}$ is a simplicial co-boundary in the ordinary sense, and not an information co-cycle.
\begin{remark}
From the quantitative point of view, we have also considered and quantified on data the pseudo-concave function $(-1)^{k-1}I_k$ (in the sense of \cite{Baudot2015a}) as a measure of available information in the total system and  considered total variation along paths. Although such functions are sounding and appealing, we have chosen to illustrate here only the results using the function $I_k$ as they respect and generalize the usual multivariate statistical correlation structures of the data and provide meaningful data interpretation of positivity and negativity, as it will become obvious in the following application to data. However, what really matters is the full landscape of information sequences, showing that information is not well described by a unique number,  but rather by a collection of numbers indexed by collections of joint variables.
\end{remark}

\subsection{Application to Gene expression data - detection of cell types and gene modules}
\label{Computation of Simplicial Information Cohomology}

The developments and tests of the estimation of simplicial information topology on data is made on a genetic expression dataset of two cell types obtained as described in the section material and methods \ref{Genetic expression data}. The result of this quantification of gene expression is represented in "Heat maps" and allows two kinds of analysis:
\begin{itemize}
	\item The analysis with genes as variables: in this case the "Heat maps" correspond to $(m,n)$ matrices $D$ (presented in the section \ref{Probability estimation}) together with the labels (population A or population B) of the cells. The data analysis consists in the detection  of gene modules. 
	\item  The analysis with cells (neurons) as variables: in this case the "Heat maps" correspond to the transposed  matrices $D^T$ (presented in the Figure \ref{figure_Supp_Result_dopa_nondopa_infopath}) together with the labels (population A or population B) of the cells. The data analysis consists in the detection of cell types. 
\end{itemize}

\subsection{Mutual-information negativity, clusters and links} \label{Mutual-information negativity}

\begin{figure} [!h]
	\centering
	\includegraphics[height=14cm]{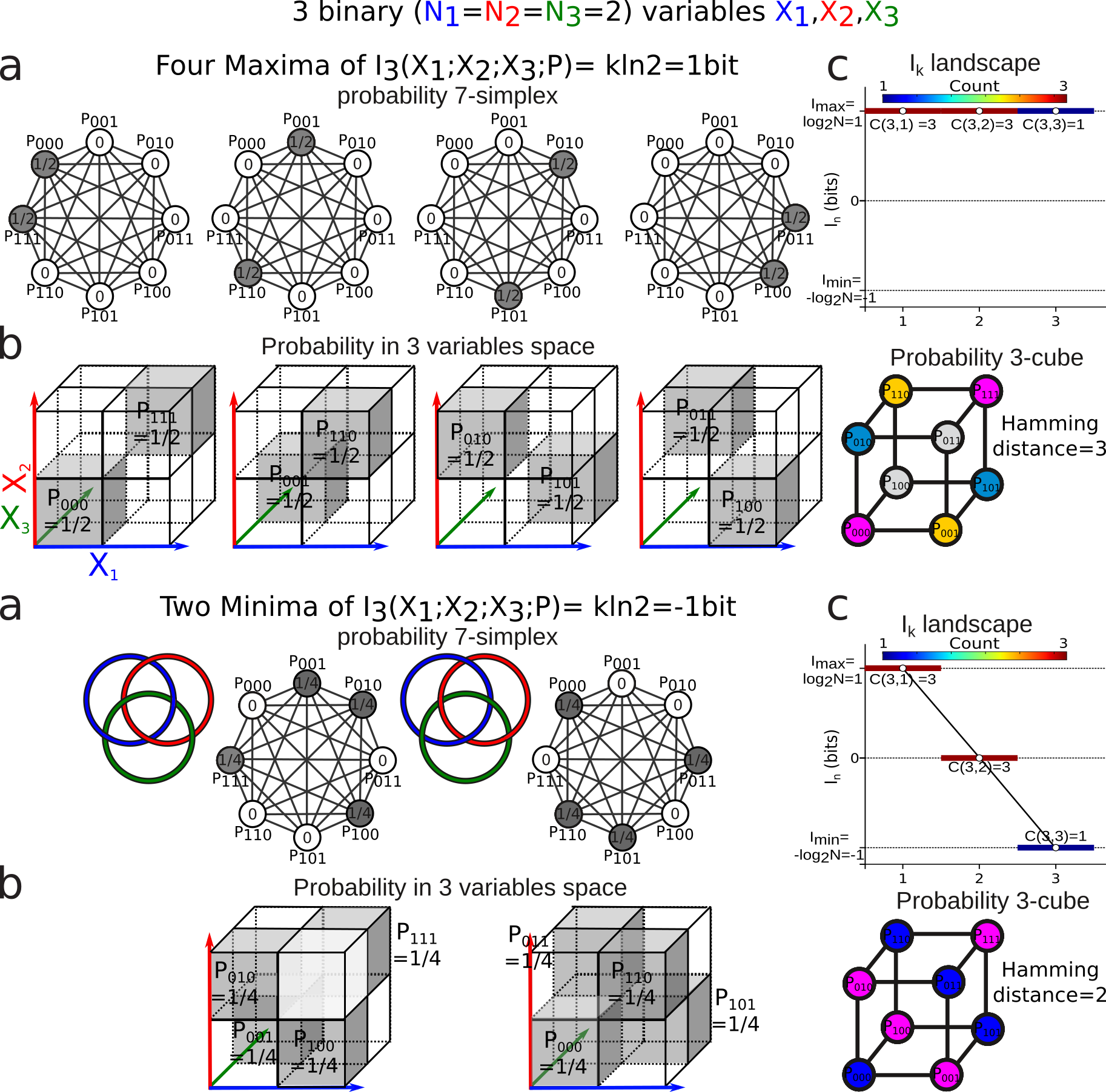}
	\caption{\textbf{Example of the 4 maxima (left panel) and of the 2 minima of $I_3$ for 3 binary variables}  \textbf{a,} informal representation of the 7-simplex of probability associated with 3 binary variables. The values of the atomic probabilities that achieve the extremal configurations are noted in each vertex. \textbf{b,} Representation of the associated probabilities in the data space of the 3-variables for these extremal configurations. \textbf{c,} Information $I_k$ landscapes of these configurations (top). Representation of these extremal configurations on the probability cube. The colors represents the non-nul atomic probability of each extremal configuration (bottom).}	,
	\label{figure_Info_negativity}
\end{figure}
As information negativity posed problems of interpretation as recalled in introduction and conclusion, we now illustrate what the negative and positive information values quantify on data with theoretical (taking the case of the binary variable case previously exposed) and empirical N-ary case examples of gene expression.
Let us consider three ordinary biased coins  $X_1, X_2, X_3$, we will denote by $0$ and $1$ their individual states and by $a,b,c,...$ the probabilities of their possible configurations three by three; more precisely:
\begin{align}
	a=p_{000},\quad b=p_{001},\quad c=p_{010},\quad d=p_{011},\\
	e=p_{100},\quad f=p_{101},\quad g=p_{110},\quad h=p_{111}.
\end{align}
We have
\begin{equation}
	a+b+c+d+e+f+g+h=1.
\end{equation}

\noindent The following identity is easily deduced from the definition of $I_3$ (cf. \ref{nonlocfundeqinfo}):
\begin{equation}
I(X_1;X_2;X_3)=I(X_1;X_2)-I(X_1;X_2|X_3).
\end{equation}
Of course the identities obtained by changing the indices are also true. This identity interprets the information shared by three variables as a measure of
the lack of information in conditioning. We notice a kind of intrication of $I_2$: conditioning can increase the information, which interprets rightly the negativity
of $I_3$. Another useful interpretation of $I_3$ is given by
\begin{equation}
I(X_1;X_2;X_3)=I(X_1;X_3)+I(X_2;X_3)-I((X_1,X_2);X_3).
\end{equation}
In this case negativity is interpreted as a synergy, i.e. the fact that two variables can
give more information on a third variable than the sum of the two separate
information.\\

\noindent Several inequalities are easy consequences of the above formulas and of
the positivity of mutual information of two variables (conditional or not), as shown in \cite{Matsuda2001}.\\
\begin{equation}
I(X_1;X_2;X_3)\leq I(X_1;X_2),
\end{equation}
\begin{equation}
I(X_1;X_2;X_3)\geq -I(X_1;X_2|X_3),
\end{equation},
and the analogs that are obtained by permuting the indices.\\
Let us remark that this immediately implies the following assertions:\\
1) when two variables are independent the information of the three is negative or zero;\\
2) when two variables are conditionally independent with respect to the third, the information of the
three is positive or zero.\\
\noindent By using the positivity of the entropy (conditional or not), we also have:
\begin{equation}
I(X_1;X_2)\leq min(H(X_1),H(X_2)),
\end{equation}
\begin{equation}\label{minoration}
I(X_1;X_2|X_3)\geq -min (H(X_1|X_3),H(X_2|X_3))\geq -min (H(X_1),H(X_2)).
\end{equation}
We deduce from here
\begin{equation}\label{firstinequality}
I(X_1;X_2;S_3)\leq min (H(X_1),H(X_2),H(X_3)),
\end{equation}
\begin{equation}\label{secondinequality}
I(X_1;X_2;X_3)\geq -min (H(X_1),H(X_2),H(X_3)).
\end{equation}
In the particular case of three binary variables, this gives
\begin{equation}
1\geq I(X_1;X_2;X_3)\geq -1.
\end{equation}
\begin{propos} 
	The absolute maximum of $I_3$, equal to $1$, is attained only in the four cases of three identical or opposite unbiased variables. That is $H(X_1)=H(X_2)=H(X_3)=1$, and $X_1=X_2$ or $X_1=1-X_2$, and
	$X_1=X_3$ or $X_1=1-X_3$, that is $a=h=1/2$ or $b=g=1/2$ or $c=f=1/2$ or $d=e=1/2$
	and in each case all the other variables are equal to $0$ (cf. Figure \ref{figure_Info_negativity} a-c).
\end{propos} 

\begin{proof}
	First it is evident that the example gives $I_3=1$. Second,
	consider three variables such that $I(X_1;X_2;X_3)=1$. We must have
	$H(X_1)=H(X_2)=H(X_3)=1$, and also $I(X_i;X_j)=1$ for any pair $(i,j)$, thus
	$H(X_i,X_j)=1$, $H(X_i|X_j)=0$, and the variable $X_i$ is a deterministic
	function of the variable $X_j$, which gives $X_i=X_j$ or $X_i=1-X_j$.
\end{proof}

\begin{propos}
	The absolute minimum of $I_3$, equal to $-1$, is attained only in the two cases of three two by two independent unbiased variables satisfying $a=1/4, b=0, c=1/4, d=0, e=1/4, f=0, g=1/4, h=0$, or $a=0,
	b=1/4, c=0, d=1/4, e=0, f=1/4, g=0, h=1/4$. These cases correspond to the two
	borromean links, the right one and the left one (cf. Figure \ref{figure_Info_negativity} d-f).
\end{propos}

\begin{proof} 
	First it is easy to verify that the examples give $I_3=-1$. Second consider three variables such that $I(X_1;X_2;X_3)=-1$. The inequality \eqref{secondinequality} implies $H(X_1)=H(X_2)=H(X_3)=1$, and the
	inequality \eqref{minoration} shows that $H(X_i|X_j)=1$ for every pair of different indices, so $H(X_1,X_2)= H(X_2,X_3)=H(X_3,X_1)=2$, and the three variables are two by two independent. Consequently the total entropy $H_3$ of $(X_1,X_2,X_3)$, given by $I_3$ minus the sum of individual entropies plus the
	sum of two by two entropies is equal to $2$. Thus
	\begin{equation}\label{sumlg}
	8=-4a\lg a-4b\lg b-4c\lg c-4d\lg d- 4e\lg e-4f\lg f-4g\lg g-4h\lg h.
	\end{equation}
	But we also have
	\begin{equation}
	8=8a+8b+8c+8d+8e+8f+8g+8h,
	\end{equation}
	that is
	\begin{equation}\label{sumprobas}
	8=4a\lg 4+4b\lg 4+4c\lg 4+4d\lg 4+4e\lg 4+4f\lg 4+4g\lg 4+4h\lg 4.
	\end{equation}
	Now we subtract \eqref{sumprobas} from \eqref{sumlg}, we obtain
	\begin{equation}\label{sumlg2}
	8=-4a\lg 4a-4b\lg 4b-4c\lg 4c-4d\lg 4d- 4e\lg 4e-4f\lg 4f-4g\lg 4g-4h\lg 4h.
	\end{equation}
	However each of the four quantities $-4a\lg 4a-4b\lg 4b,-4c\lg 4c-4d\lg 4d,\\-
	4e\lg 4e-4f\lg 4f,-4g\lg 4g-4h\lg 4h$ is $\geq0$ because each of the four sums
	$4a+4b,4c+4d,4e+4f,4g+4h$ is equal to $1$, so each of these quantities is equal
	to zero, which happens only if $ab=cd=ef=gh=0$. But we can repeat the argument
	with any permutation of the three variables $X_1,X_2,X_3$. We obtain nothing
	new from the transposition of $X_1$ and $X_3$. From the transposition of $X_1$
	and $X_3$ we obtain $ae=bf=cg=dh=0$. From the transposition of $X_2$ and $X_3$,
	we obtain $ac=bd=eg=fh=0$. So from the cyclic permutation $(1,2,3)$ (resp.
	$(1,3,2)$, we get $ae=bf=cg=dh=0$ (resp. $ac=bd=eg=fh=0$).) If $a=0$ this gives
	necessarily $b,e,c$ nonzero, thus $d=f=g=0$, and $h\neq 0$, and if $a\neq 0$
	this gives $b=e=c=0$, thus $d,f,g$ nonzero and $h=0$.
\end{proof}

\noindent Figure \ref{figure_Info_negativity} illustrates the probability configurations giving rise to the maxima and minima of $I_3$ for 3 binary variables.\\

\noindent In the much more complex case of gene expressions, the statistical analysis shown in \cite{Tapia2018} exhibited also a combination of positivity and negativity of the information quantities $I_k;k\geq 3$. In this analysis, the minimal negative information configurations provide a clear example of purely emergent and collective interactions analog to Borromean links in topology, since it cannot be detected from any pairwise investigation or 2-dimensional observations. In these Borromean links the variables are pairwise independent but dependent at 3. In general $I_k$ negativity detects such effects of their projection on lower dimensions, this illustrates the main difficulty when going from dimension $2$ to $3$ in information theory. The example given in Figure \ref{figure_Info_negativity} provides a simple example of this dimensional effect in the data space: the alternated clustering of the data corresponding to $I_3$ negativity cannot be detected by the projections onto whichever subspace of pair of variables, since the variables are pairwise independent. For N-ary variables the negativity becomes much more complicated, with more degeneracy of the minima and maxima of $I_k$.

\begin{figure} [!h]
	\centering
	\includegraphics[height=9.5cm]{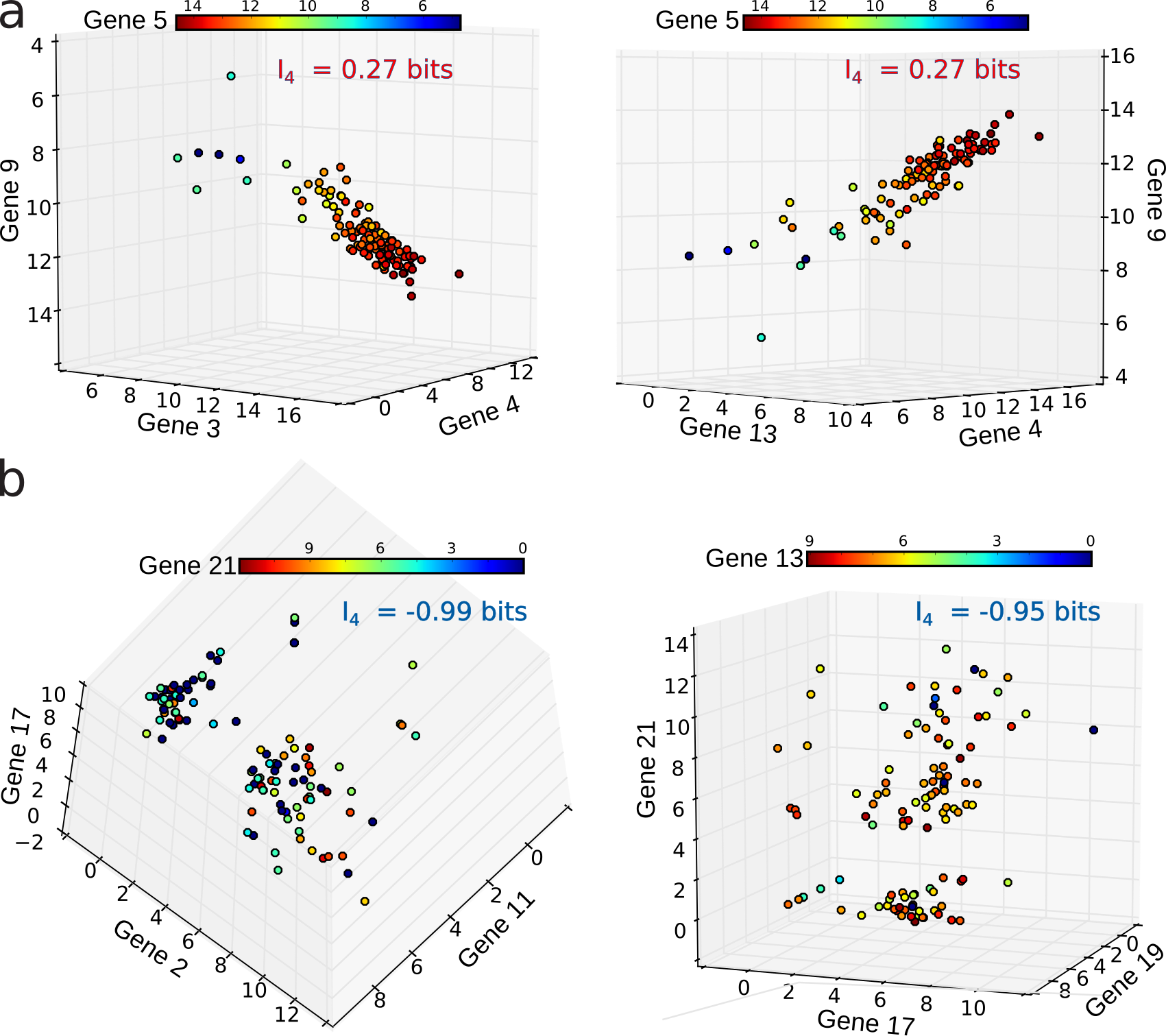}
	\caption{\textbf{Examples of some of 4-modules (quaduplets) with the highest (positive) and lowest (negative) $I_4$ of gene expression represented in the data space. }  \textbf{a,} Two 4-modules of genes sharing among the highest positive $I_4$ of the gene expression data set (cf. \ref{Genetic expression data}). The data are represented in the data space of the measured expression of the 4 variables-genes. The fourth dimension-variable is color coded. \textbf{b,}  Two 4-modules of genes sharing among the lowest negative $I_4$. All the modules were found to be significant according to the dependence test introduced in section \ref{k-independence test}, except the module $\{17,19,21,13\}$. The identified extremal modules (different) give similar patterns of dependences \cite{TapiaPacheco2017,Tapia2018}.}
	\label{figure_negativity_positivity}
\end{figure}

In order to illustrate the theoretical examples of Figure \ref{figure_Info_negativity} on real data, considering the data set of gene expression (matrix $D$), we plotted some quadruplets of genes sharing some of the highest (positive) and lowest (negative) $I_4$ values in the data space of the variables (Figure \ref{figure_negativity_positivity}). Figure \ref{figure_negativity_positivity} shows that in the data space, $I_k$ negativity identifies the clustering of the data points, or in other words, the modules (k-tuples) for which the data points are segregated into condensate clusters.  As expected theoretically, $I_k$ positivity identifies co-variations of the variables, even in cases of non-linear relations, as shown by Reshef and colleagues \cite{Reshef2011} in the pairwise case. It can be easily shown in the pairwise case that $I_k$ positivity generalizes the usual correlation coefficient to non-linear relations.  As a result, the interpretation of the negativity of $I_k$ is that it provides a signature and quantification of the variables that segregate or differentiate the measured population.

\subsection{Cell type detection - comparison with previous MaxEnt studies}

\paragraph{Example of cell type recognition with a low sample size $m=41$, dimension $n=20$, and graining $N=9$.} \label{Example of landscapes}

\begin{figure} [!h]
	\centering
	\includegraphics[height=17.7cm]{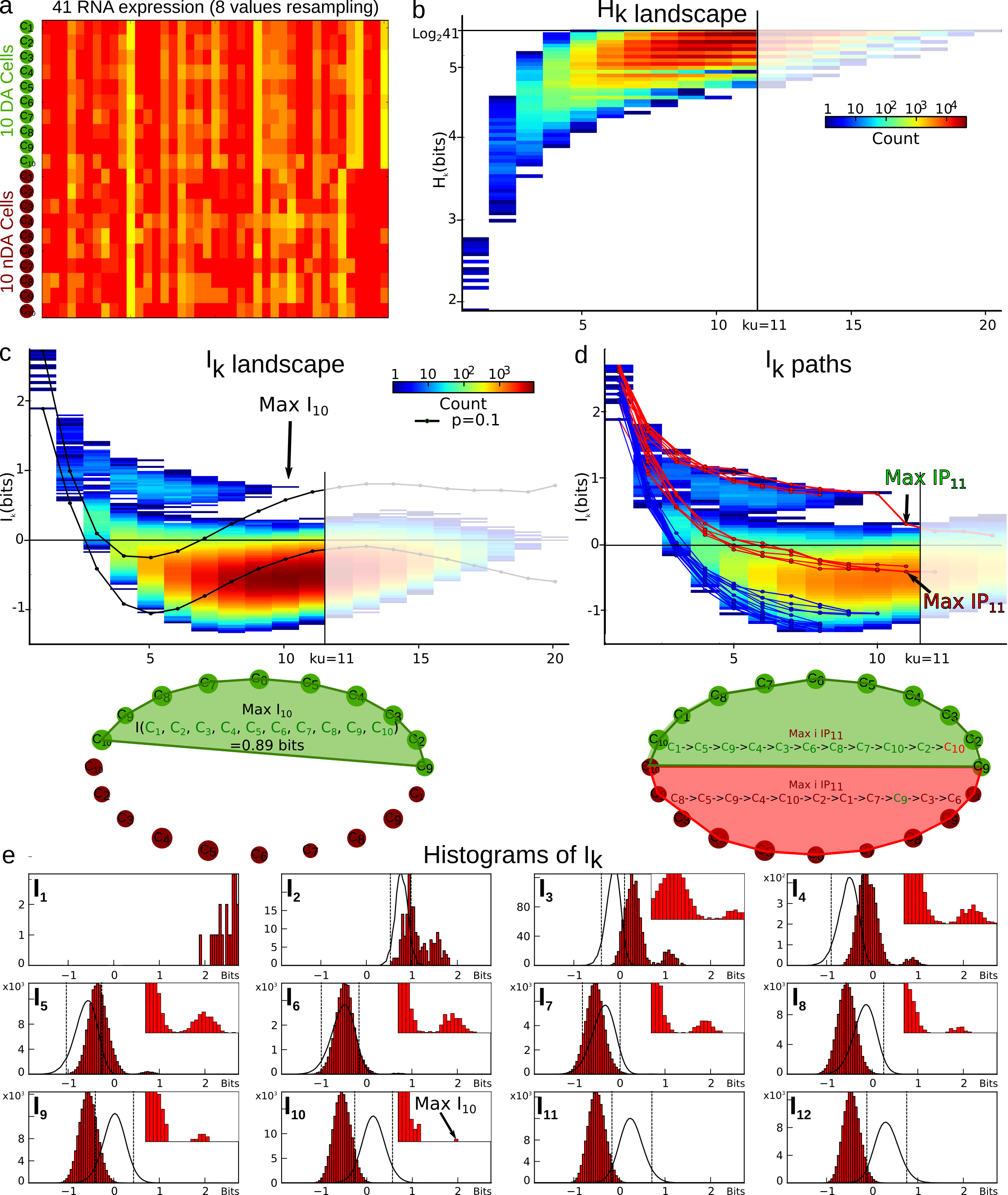}
	\caption{\textbf{Example of a $I_k$ landscape and path analysis. a,} heatmap (transpose of matrix $D$) of $n=20$ neurons with $m=41$ genes.\textbf{b,} the corresponding $H_k$ landscape. \textbf{c,} the corresponding $I_k$ landscape \textbf{d,} maximum (in red) and minimum (in blue) $I_k$ information paths. \textbf{e,}  histograms of the distributions of $I_k$ for $k=1,..,12$. See text for details.}
	\label{figure_Supp_Result_dopa_nondopa_infopath}
\end{figure}

As introduced in previous section \ref{Mutual-information negativity}, the k-tuples presenting the highest and lowest information ($I_k$) values are the most biologically relevant modules and identify the variables that are the most dependent or synergistic (respectively "entangled"). We call information landscape the representation of the estimation of all $I_k$ values for the whole simplicial lattice of k-subfaces of the $n$-simplex of variables ranked by their $I_k$ values in ordinate.  In general the null hypothesis against whom are tested the data is the maximal uniformity and independence of the variables $X_i,i=1,...,n$. Below the undersampling dimension $k_u$ presented in methods \ref{Undersampling dimension}, this predicts the following standard sequence for any permutation of the variables $X_{i_1},...,X_{i_n}$;
\begin{equation}
H_1=\log_2 r,..., H_k=k\log_2 r,...
\end{equation}
that is linearity (with $N_1=...=N_n=r$).\\
What we observed in the case where independence is confirmed, for instance with the chosen genes of the population B (NDA neurons) in \cite{Tapia2018}, is linearity up to the maximal significant $k$, then stationarity. But where independency is violated, for example with the chosen genes of the population A (DA neurons) in \cite{Tapia2018}, some permutations of $X_1,...,X_n$ give sequences showing strong departures from the linear prediction.\\
This departure and the rest of the structure can also be observed on the sequence $I_k$ as shown in Figure  \ref{figure_Supp_Result_dopa_nondopa_infopath} and \ref{applied_TIDA}, which present the case where cells are considered as variables. In the trivial case, i.e. uniformity and independence, for any permutation, we have
\begin{equation}
I_1=\log_2 r, I_2=I_3=...=I_n=0.
\end{equation}

As detailed in materials and methods \ref{Computation_minimum_free_energy_complex},  we further compute the longest information paths (starting at 0 and that go from vertex to vertex following the edges of the simplicial lattice) with maximal or minimal slope (with minimal or maximal conditional mutual-information) that end at the first minimum, a conditional-independence criterion (a change of sign of conditional mutual-information). Such paths select the biologically relevant variables that progressively add more and more dependences. \noindent The paths $I_k(\sigma)$ that stay strictly positive for a long time are especially interesting, being interpreted as the succession of variables $X_{\sigma_1},...,X_{\sigma_k}$ that share the strongest dependence. \noindent However, the paths $I_k(\sigma)$ that become negative for $k\geq 3$ through  $I_2\approx 0$ are also interesting, because they exhibit a kind of frustration in the sense of Matsuda \cite{Matsuda2001} or synergy in the sense of Brenner \cite{Brenner2000a}.\\

The information landscape and path analysis corresponding to the analysis with cells as variables are illustrated in Figure \ref{figure_Supp_Result_dopa_nondopa_infopath}. It comes to consider the cells as a realization of gene expression rather than the converse, cf. \cite{Dawkins1976}. In this case, the data analysis task is to recover blindly the pre-established labels of cell types (population A and population B) from the topological data analysis, an unsupervised learning task. The heat-map transpose matrix of $n=20$ cells with $m=41$ genes is represented in Figure \ref{figure_Supp_Result_dopa_nondopa_infopath}a.  We took $n=20$ neurons among the $148$ within which $10$ were pre-identified as population A neurons (in green) and $10$ were pre-identified as population B neurons (in dark red), and ran the analysis on the $41$ gene expression with a graining of $N=9$ values (cf. section \ref{Genetic expression data}). The dimension above which the estimation of information becomes too biased due to the finite sample size is given by the undersampling dimension $k_u=11$ (p value 0.05, cf. section \ref{Undersampling dimension}). The landscapes turn out to be very different from the extremal (totally disordered and totally ordered) homogeneous (identically distributed) theoretical cases. The $I_k$ landscape shown in Figure \ref{figure_Supp_Result_dopa_nondopa_infopath}c exhibits two clearly separated components. The scaffold below represents the tuple corresponding to the maximum of $I_{10}$: it corresponds exactly to the $10$ neurons pre-identified as being population A neurons.\\

The  maximum (in red) and minimum (in blue) $I_k$ information paths identified by the algorithm are represented in Figure \ref{figure_Supp_Result_dopa_nondopa_infopath}d. The scaffold below represents the two tuples corresponding to the two longest maximum paths in each component: the longest (noted Max $IP_{11}$ in green) $IP_{11}$ contains the 10 neurons pre-identified as population A and 1 "error" neuron pre-identified as population B. We restricted the longest maximum path to the undersampling dimension $k_u=11$, but this path reached $k=14$ with erroneous classifications. The second longest maximum path (noted Max $IP_{11}$ in red) $IP_{11}$ contains the 10 neurons pre-identified as population B and 1 neuron pre-identified as population A that is hence erroneously classified by the algorithm. Altogether the information landscape shows that population A neurons constitute a quite homogenous population, whereas the population B neurons correspond to a more heterogeneous population of cells, a fact that was already known and reported in the biological studies of these populations. The histograms of the distributions of $I_k$ for $k=1,..,12$, shown in Figure \ref{figure_Supp_Result_dopa_nondopa_infopath}e are clearly bimodal and the insets provide a magnification on the population A component. As detailed in the section materials and methods \ref{k-independence test}, we developed a test based on the random shuffles of the data points that leave the marginal distributions unchanged, as proposed by \cite{Pethel2014}. It estimates if a given $I_k$ significantly differs from a randomly generated $I_k$, a test of the specificity of the k-dependence. The shuffled distributions and the significance value for $p=0.1$ are depicted by the black lines and the doted lines, as in Figure \ref{figure_independencetest}. As illustrated in the histograms of Figure  \ref{figure_Supp_Result_dopa_nondopa_infopath}e and in \cite{TapiaPacheco2017}, these results show that higher dependences can be important but they do not mean that pairwise or marginal informations are not: the consideration of higher dependences can only improve the efficiency of the detection obtained from pairwise or marginal considerations. \\

\begin{figure} [!h]
	\centering
	\includegraphics[height=8cm]{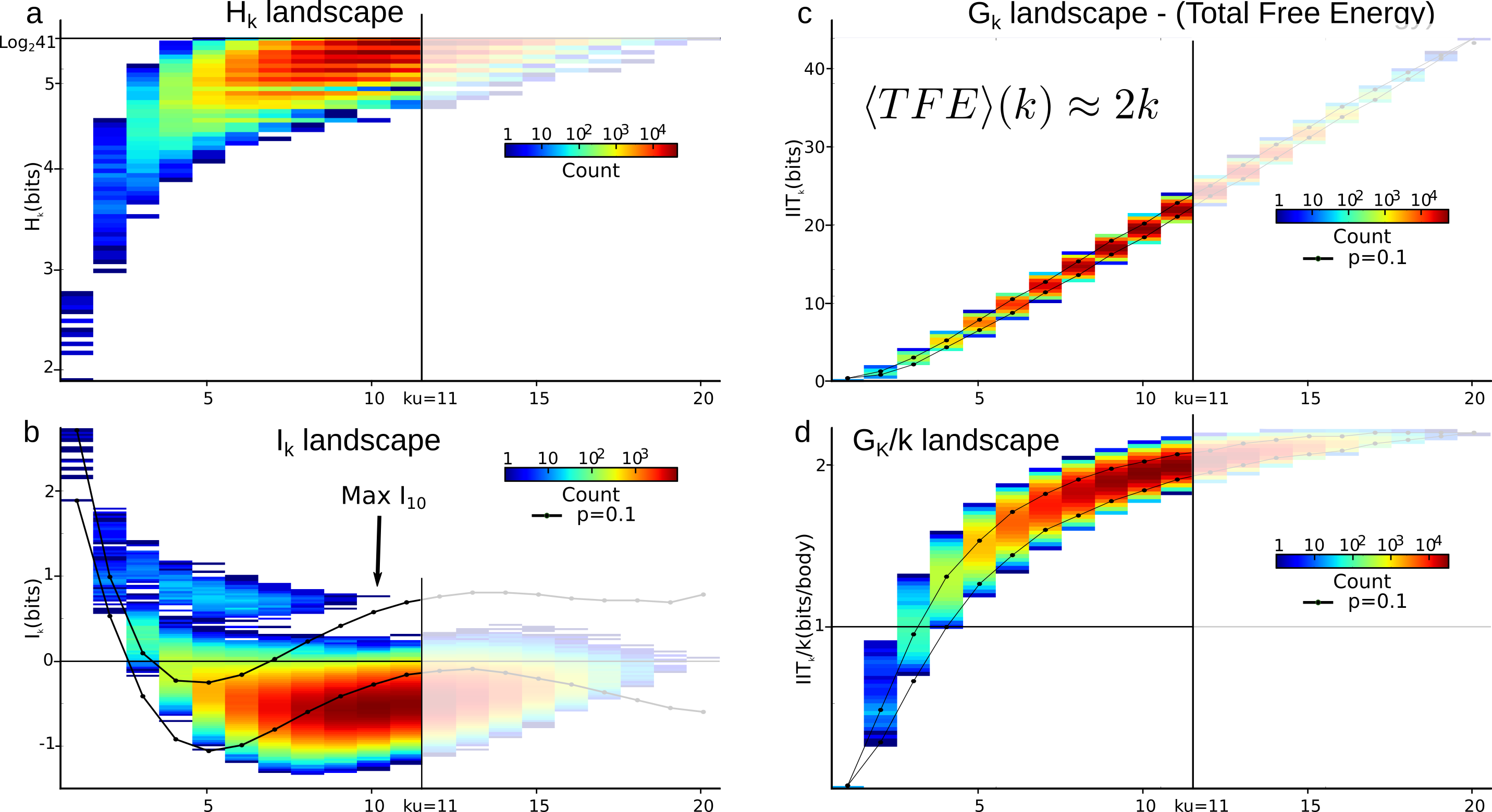}
	\caption{\textbf{$I_k$, $H_k$ and $G_k$ (Total Free Energy, TFE) landscapes}. \textbf{a:} entropy  $H_k$ and \textbf{b:} mutual information $I_k$ (free energy components) landscapes (same representation as figure \ref{figure_Supp_Result_dopa_nondopa_infopath}, $k_u = 11$, p value 0.05).\textbf{c:} $G_k$  landscape (total correlation or multi-information or Integrated Information or total free energy) \textbf{d:} the landscape of the $G_k$ per body ($G_k/k$). }	
	\label{applied_TIDA}
\end{figure}

As illustrated in Figure \ref{applied_TIDA} and expected from relative entropy positivity, the total correlation $G_k$ (see appendix \ref{Appendix: Bayes free energy} on bayes free-energy) is  monotonically increasing with the order $k$, and quite linearly in appearance ($G_k \approx 2 k$ asymptotically). The panel \textbf{d} quantifies this departure from linearity. However the  $G_k$ landscape fails to distinguih as clearly as $I_k$ landscape does, the population A.

\section{Discussion}

During the last decades, there have been important efforts in trying to evaluate the pairwise and higher order interactions in neuronal and biological measurements, notably to extract the undergoing collective dynamics. Applying the Maximum of Entropy principle on Ising spin models to neural data \cite{Schneidman2006,Tkacik2014}, a first series of studies concluded that pairwise interactions are mostly sufficient to capture the global collective dynamics, leading to the "pairwise sufficiency" paradigm (see Merchan and Nemenman for presentation \cite{Merchan2016}). However, as shown by the Ising model itself, near a second order phase transition, elementary pairwise interactions are compatible with non-trivial higher-order dependences, and very large correlations at long distances. From the mathematical and physical point of view, this fact is nicely encoded in the normalization factor of the Boltzmann probability, namely the Partition Function $Z(\beta)$. As shown by the Ising model, the probability law can be factorized (up to the normalization number $Z$) on the edges and vertices of a graph, but the statistical clusters can have unbounded volumes. Moreover, subsequent studies notably of Tka\u{c}ik et al. \cite{Tkacik2014a} (see also \cite{Humplik2017}) have shown that for sufficiently large populations of registered neurons, the pairwise models are insufficient to explain the data as proposed in \cite{Atick1992,Baudot2006} for example. Thus the dimension of the interactions to be taken into account for the models must be larger than two.\\

Note that most interactions in Biology are nowadays described in terms of networks, such that the concepts of protein networks, genetic networks or neural networks became familiar. However from the physical as well as the biological  point of view, none of these systems are really $1$-dimensional graphs, and it is now clear for most researchers in the domain that higher order structures are needed for describing collective dynamics, cf. for instance \cite{Yedidia2001}, and \cite{Reimann2017}.  Our Figure \ref{Figure_Ik_Hk_landscape}  clearly shows this point.\\
\indent William Bialek and his collaborators have well explained the interest of a systematic study of joint entropies and general multi-modal mutual information quantities as an efficient way for understanding neuronal activities, networks of neurons, and gene expression \cite{Brenner2000a,Schneidman2003b,Slonim2005}). They also developed approximate computational methods for estimating the information quantities. Mutual information analysis was applied for linking adaptation to the preservation of the information flow \cite{Brenner2000b,Laughlin1981}. Closely related to the present study,  Margolin, Wang, Califano and Nemenman have investigated multivariate dependences of higher order \cite{Margolin2010} with MaxEnt methods, by using the total-correlation $G_k$ (cf. equation \ref{total correlation}) in function of the integer $k\geq 2$. The apparent advantage is the positivity of the $G_k$.\\\\
In this respect, the originality of our method relies first on the systematic consideration of the entropy paths and the information paths that can be associated to all possible permutations of the basic variables, in arbitrary dimension, and the extraction of exceptional paths from them, in order to define the overall form of the distribution of information among the set of variables. Secondly, we used and proved the relevance of peculiar functions, multivariate mutual informations, where the previously cited works focused on total correlations, which fail to uncover the data structure as exemplified in Figure \ref{applied_TIDA} or only explored pairwise or $I_3$.
We named these tools the information landscapes of the data. This new perspective and mathematical justification of these functions has its origin in the local (topos) homological theory introduced in \cite{Baudot2015a} developped and extended in several ways by Vigneaux \cite{Vigneaux2019}. In the present article, we also proved new theoretical results along this line, about the concrete structure of higher-order information functions. Moreover the method was successfully applied to a concrete problem of gene expression in \cite{Tapia2018}.\\

Since their introduction, the possible negativity of the $I_k$ functions for $k\geq 3$ has posed serious problems of interpretation, and it was the main argument for many theoretical studies to discard such a family of functions for measuring information dependences and statistical interactions. Notably, it motivated the proposition of non-negative decomposition by Williams and Beer \cite{Williams2010} and of "unique information" by Bertschinger and colleagues \cite{Olbrich2015,Bertschinger2014}, or Griffith and Koch \cite{Griffith2014}. These partial decompositions of information are the subject of several recent investigations notably with  applications to the development of neural network \cite{Wibral2017} and neuromodulation \cite{Kay2017}. However, Rauh and colleagues showed that no non-negative decomposition can be generalized to multivariate cases for degrees higher than 3 \cite{Rauh2014} (th.2).\\
The present paper and \cite{Tapia2018} show that, to the contrary, the possible negativity is an advantage. The interest of this negativity was already illustrated in \cite{Brenner2000a,Schneidman2003b,Matsuda2001,Kim2010}, but we have further developed this topic with the study of complete $I_k$-landscapes, providing some new insights with respect to their meaning in terms of data point clusters.\\
\indent The precise contribution of higher-order is indeed directly quantified by the $I_k$ values in the landscapes and paths. Figure \ref{Figure_Ik_Hk_landscape} further illustrates the gain and the importance of considering higher statistical interactions, using the previous example of cells pre-identified as 10 population A and 10 population B cells ($n=20$, $m=47$, $N=9$). The plots are the finite and discrete analogs of Gibbs's  original representation of entropy vs. energy \cite{Gibbs1873}. Whereas pairwise interactions ($k=2$) cannot (or very hardly) distinguish the population A and population B cell types, the maximum of $I_{10}$ unambiguously identifies the population A.
	\begin{figure} [!h]
		\centering
		\includegraphics[height=6cm]{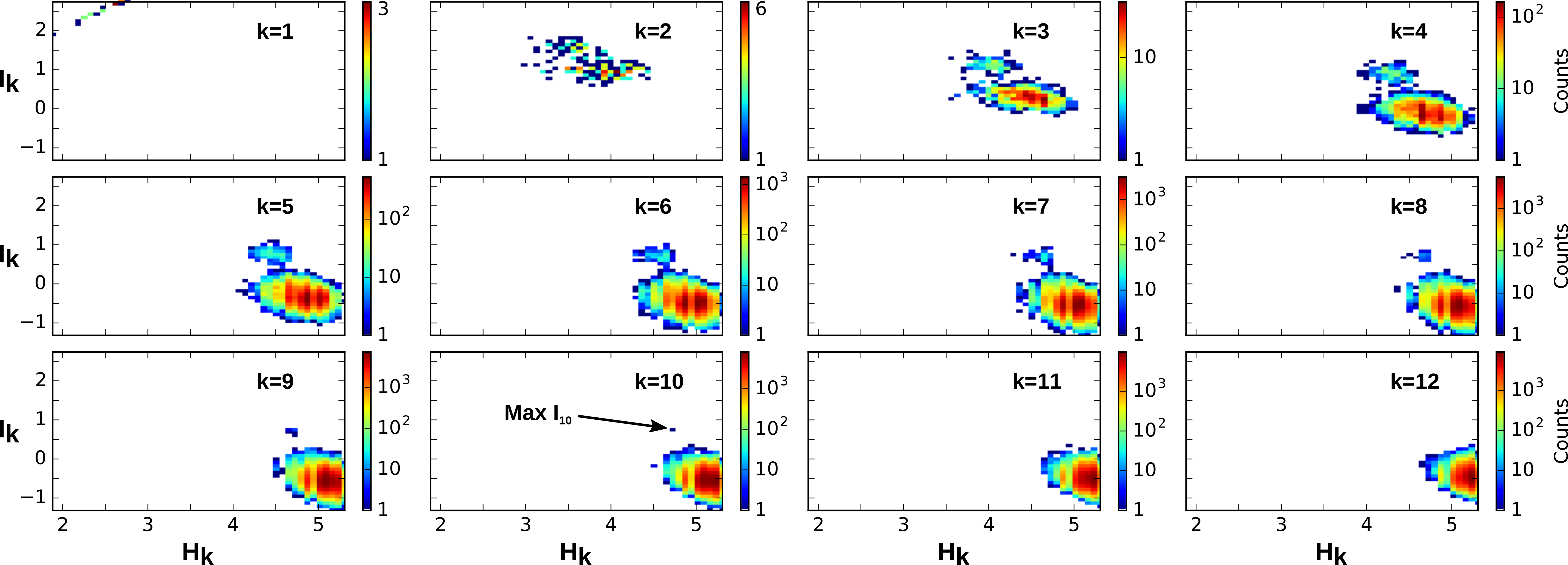}
		\caption{
			\textbf{$H_k-I_k$ landscape: Gibbs-Maxwell's entropy vs. energy representation.} $H_k$ and $I_k$ are plotted in abscissa and ordinate respectively for dimension $k=1,...,12$ for the same data and setting as in Figure \ref{figure_Supp_Result_dopa_nondopa_infopath} ($n=20$ cells, $m=47$ genes, $N=9$, $k_u=11$). Compare the difficulty in identifying the 2 cells types from the pairwise $k=2$ landscape to the $k=10$ landscape.}
		\label{Figure_Ik_Hk_landscape}
	\end{figure} 
	
\indent As illustrated in Figure \ref{figure_Supp_Result_dopa_nondopa_infopath}, the present analysis shows that in the expression of 41 genes of interest of population A neurons, the higher-order statistical interactions are non-negligible and have a simple functional meaning of a collective module, a cell type. We believe such conclusion to be generic in biology. More precisely, we believe that even if related to physics, biological structures have higher-order statistical interactions defined by higher-order information and that these interactions provide the signature of their memory engramming. In fact "information is physical" as stated by Bennett following Landauer \cite{Landauer1961}), in the sense of memory capacity and necessity of forgetting.  The quantification of the information storage applied here to genes can be considered as a generic epigenetic memory characterization, resulting of a developmental-learning process. The consideration of higher dimensional statistical dependences increases combinatorially the number of possible information modules engrammed by the system. It hence provides an appreciable capacity reservoir for information storage and for differentiation, for diversity. \\
\indent The critical points of the Ising model in dimension 2 and 3 show the difficulty to relate factorization (up to $Z(\beta)$), which describes the manner energy interactions localize, with the dependences structure, or in other words the manner information distributes itself, i.e. the form of information. Only few theoretical results relate the two notions. However, on the basis of several recent studies that we mentioned, particularly the studies of adaptive functions, and comforted by the analysis presented in this article, we can suggest that for biological systems, during development or evolution, the distribution of the information flow, as described in particular by higher order information quantities, participates in the generators of the dynamics, on the side of energy quantities coming from Physics. \\

\section{Material and Methods} \label{Material and Methods}

\subsection{The dataset: quantified genetic expression in two cell types} \label{Genetic expression data}

The quantification of genetic expression was performed using microfluidic qPCR technique on single dopaminergic (DA) and non-dopaminergic (NDA) neurons isolated from two midbrain structures, the Substantia Nigra pars compacta (SNc) and the neighboring Ventral Tegmental Area (VTA), extracted from  adult TH-GFP mice (transgenic mice expressing the Green Fluorescent Protein under the control of the Tyrosine Hydroxylase promoter). The precise protocols of extraction, quantification, and identification are detailed in \cite{TapiaPacheco2017,Tapia2018}. This technique allowed us to quantify in a single cell the levels of expression of 41 genes chosen for their implication in neuronal activity and identity of dopaminergic (DA) neurons. The SNc DA neurons were identified based on GFP fluorescence (TH expression). This identification was further confirmed based on the expression levels of \textit{Th} and \textit{Slc6a3} genes, which are established markers of DA metabolism. The quantification of the expression of the 41 genes ($n=41$) was achieved in 111 neurons ($m=111$) identified as DA and in 37 neurons ($m=37$) identified as nDA. In this article, for readability purpose, we replaced the names of the genes by gene numbers and the cell type DA by population A, and the cell type nDA by population B. The dataset is available in supplementary material \cite{TapiaPacheco2017,Tapia2018}.

\subsection{Probability estimation}\label{Probability estimation}

The presentation of the probability estimation procedure is achieved on matrices $D$ (genes as variables), and it is the same in the case of the analysis of the matrices $D^T$ (cells as variables). It is illustrated in Figure \ref{figure_Supp_proba_estimation} for the simple case of 2 random variables taken from the dataset of gene expression presented in section \ref{Genetic expression data}, namely the expression of two genes Gene5 and Gene21 in $m=111$ population A cells. Our probability estimation corresponds to a step of the integral estimation procedure of Riemann.

\begin{figure} [!h]
	\centering
	\includegraphics[height=12.5cm]{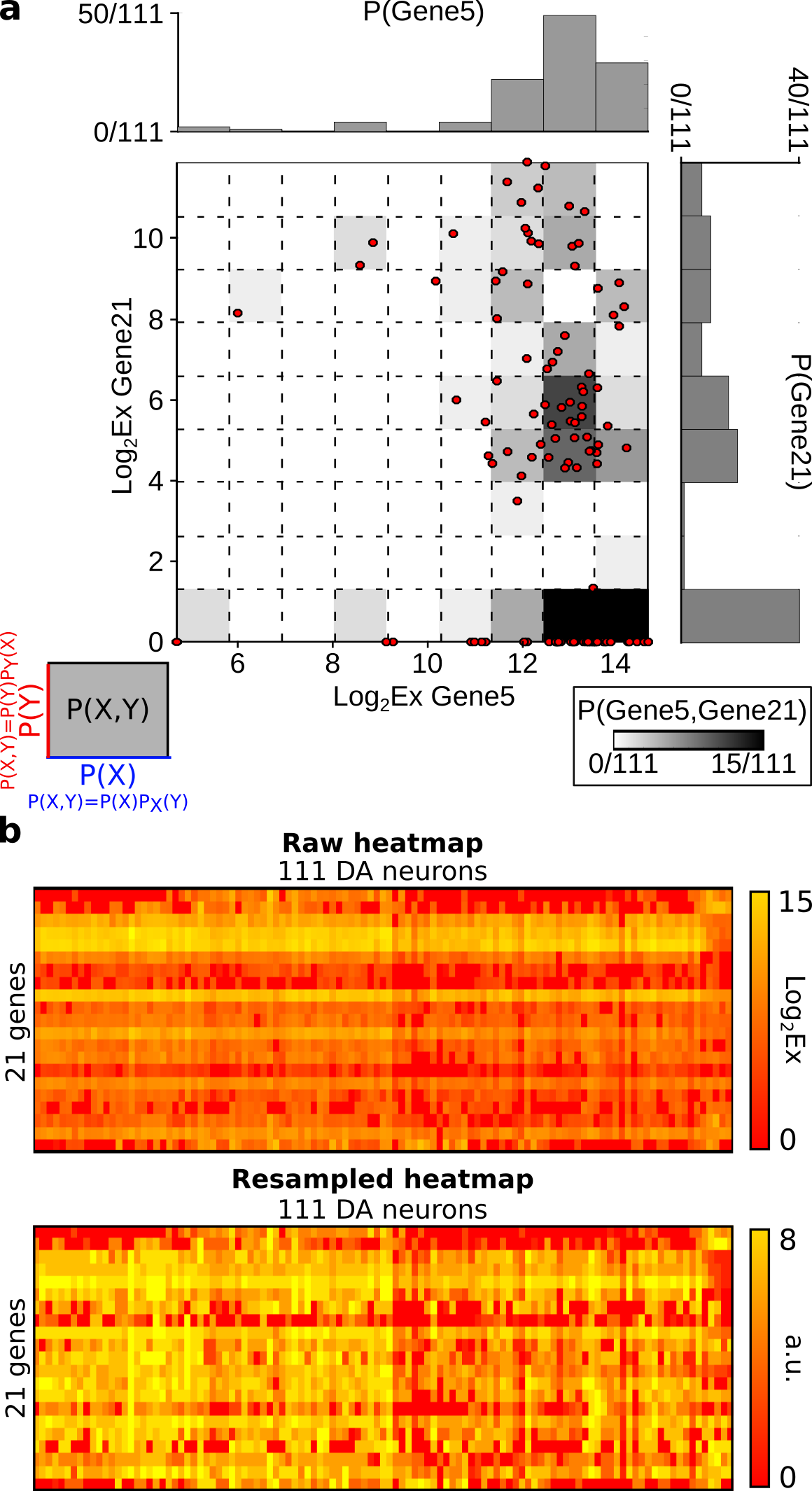}
	\caption{\textbf{Principles of probability estimation for 2 random variables.} a, illustration of the basic procedure used in practice to estimate the probability densitiy for the two genes ($n=2$) Gene5 and Gene21 in 111 population A neurons ($m=111$) using a graining of 9 ($N_1=N_2=9$). The data points corresponding to the 111 observations are represented as red dots, and the graining is depicted by the 81-box grid ($N_1.N_2$). The borders of the graining interval are obtained by considering the maximum and minimum measured values for each variable, and data are then sampled regularly within this interval with $N_i$ values. Projections of the data points on lower dimensional variable subspaces ($X_1$ and $X_2$ axes here) are obtained by marginalization, giving the marginal probability laws for the 2 variables $X_1$ and $X_2$ ($P_{X_i,N_i,m}$) ; represented as histograms above the $X_1$-axis for Gene21 and on the right of the $X_2$-axis for Gene21). b, heatmaps representing the levels of expression of the 21 genes of interest on a $\log_2 Ex$ scale (top, raw heatmap) and after resampling with a graining of 9 (bottom, $N_1=N_2 =...=N_{21} =9$). }
	\label{figure_Supp_proba_estimation}
\end{figure}

We write the heatmap as a $(m,n)$ matrix $D$ and its real coefficients $x_{ij} \in \mathbb{R}, ~ i \in \{1..m\}, ~ j \in \{1...n\}$: the columns of $D$ span the $m$ repetitions-trials (here the $m$ neurons) and the rows of $D$ spans the $n$ variables (here the $n$ genes). We also note, for each variable $X_j$, the minimum and maximum values measured as $\min x_j=\min_{1 \leq i \leq m} x_{ij}$ and $\max x_j=\max_{1 \leq i \leq m} x_{ij}$. 

We consider the space in the intervals $[\min x_j,\max x_j]$ for each variable $X_j$ and divide it into $N_1.N_2...N_n$ boxes, on which it is possible to estimate the atomic probabilities by elementary counting. We note each $n$-dimensional box by an $n$-tuple of integers $\{a_1,...,a_n\}$ where $\forall i \in \{1,...,n\}, ~ a_i \in \{1,...,N_i\}$, and writing the min and the max of a box on each variable $X_j$ (the jth co-ordinate of the vertex of the box)  as $\operatorname{bmin}_j= \min x_j +\frac{(a_j-1)(\max x_j-\min x_j)}{N_j}$ and  $\operatorname{bmax}_j= \min x_j +\frac{(a_j)(\max x_j-\min x_j)}{N_j}$, then the atomic probabilities can be defined using Dirac function $\delta$ as: 

\begin{multline}\label{probability}
P\left(\operatorname{bmin}_1 \leq X_1 \leq  \operatorname{bmax}_1,\operatorname{bmin}_2 \leq X_2 \leq  \operatorname{bmax}_2,..., \operatorname{bmin}_n \leq X_n \leq  \operatorname{bmax}_n \right)\\
=\sum_{i=1}^{m} \frac{\delta_i}{m}, 
~ \delta_i =\begin{cases}
0, ~ \text{if} \operatorname{bmin}_1 > x_{i1} ~ \text{or} ~  x_{i1}>\operatorname{bmax}_1 ...\text{or} ~  \operatorname{bmin}_n > x_{in} ~ \text{or} ~  x_{in}>\operatorname{bmax}_n\\
1, ~ \text{if} \operatorname{bmin}_1 \leq x_{i1}\leq \operatorname{bmax}_1 ~ \text{and}...\text{and} ~ \operatorname{bmin}_n\leq x_{in}\leq \operatorname{bmax}_n
\end{cases}
\end{multline} 

For two variables, using the definition of conditioning $P_{X}(Y)= \frac{P(X.Y)}{P(X)}$ and in the general case using the theorem of total probability \cite{Kolmogorov1933a} ($P(X)=\sum_{i=0}^N P(A_i.X)=\sum_{i=0}^N P(A_i).P_{A_i}(X)$), we can marginalize, or geometrically project on lower dimensions, to obtain all the probabilities corresponding to subsets of variables, as illustrated in Figure \ref{figure_Supp_proba_estimation}. For example, with short notation, the  probability associated to the marginal variable $X_i$ being in the interval $[\operatorname{bmin}_i,\operatorname{bmax}_i]$ is obtained by direct summation:
\begin{multline}\label{marginalization}
P\left(\operatorname{bmin}_i \leq X_i \leq  \operatorname{bmax}_i \right)=\\
\sum_{i=1}^{N_1...\widehat{N_i}...N_n} P\left(\operatorname{bmin}_1 \leq X_1 \leq  \operatorname{bmax}_1,\operatorname{bmin}_2 \leq X_2 \leq  \operatorname{bmax}_2,..., \operatorname{bmin}_n \leq X_n \leq  \operatorname{bmax}_n \right)
\end{multline} 
In the example of Figure \ref{figure_Supp_proba_estimation}, the probability of the level of \textit{Th} being in the 4th box is:
\begin{multline}\label{marginalization_example}
P\left(8 \leq \text{Th} \leq  9.8 \right) =\\ \sum_{i=0}^{8}P\left(8 \leq \text{Th} \leq  9.8,\operatorname{bmin}_2 \leq \text{Calb1} \leq  \operatorname{bmax}_2 \right)\\
= 2/111 + 2/111
\end{multline}

In geometrical terms, the set of total probability laws is an $N=N_1.N_2...N_n-1$ dimensional simplex $\Delta_{N_1.N_2...N_n-1}$ (the $-1$ accounts for the normalization equation $\sum P_i=1$, which embeds the simplex in an affine space). In the example of Figure \ref{figure_Supp_proba_estimation}, we have an 80-dimensional probability simplex $\Delta_{80}$, the set of sub-simplicies over the k-faces of the simplex $\Delta_n$, for every $k$ between $0$ and $n$,  represents the boolean algebra of the joint-probabilities, which is equivalent in the finite case to their sigma-algebra.
In our analysis, we have chosen $N_1=N_2=...=N_n=9$ and this choice is justified in section \ref{Sampling size} using Reshef and colleagues criterion \cite{Reshef2011} and undersampling constraints.\\

In summary, our probability estimation and data analysis depend on $n$ (the number of random variables), on $m$ (the number of observations), and on $N_1,...,N_i$ (the graining). The merit of this method is its simplicity (few assumptions, no priors on the distributions) and low computational cost. There exist different methods that can significantly improve this basic probability estimation, but we leave this for future investigation. The graining given by the numbers $N=N_1.N_2...N_n$ and the sample size $m$ are important parameters of the analysis explored in this section.\\

\subsection{Computation of k-Entropy, k-Information landscapes and paths} \label{Computation_minimum_free_energy_complex}

The computational exploration of the simplicial sublattice has a complexity in $\mathcal{O}(2^n)$ ($2^n=\sum_{k=1}^n \binom{n}{k}$). In this simplicial setting we can exhaustively estimate information functions on the simplicial information structure, that is joint-entropy $H_k$ and mutual-informations $I_k$ at all dimensions $k\leq n$ and for every $k$-tuple, with a standard commercial personal computer (a laptop with processor Intel Core i7-4910MQ CPU @ 2.90GHz $\times$ 8, even though the program  currently uses only one CPU) up to $k=n=21$ in a reasonable time ($\approx 3$ hours). 
Using the expression of joint-entropy (equation \ref{higher entropy}) and the probability obtained using equation \ref{probability} and marginalization, it is possible to compute the joint-entropy and marginal entropy of all the variables. The alternated expression of n-mutual information given by equation \ref{higher information} then allows a direct evaluation of all these quantities. The definitions, formulas and theorems are sufficient to obtain the algorithm. We moreover provide the Information Topology program INFOTOPO-V1.2 under opensource licence on github depository at https://github.com/pierrebaudot/INFOTOPO. Information Topology is a program written in Python (compatible with Python 3.4.x), with a graphic interface built using TKinter \cite{Shipman2010}, plots drawn using Matplotlib \cite{Hunter2007}, calculations made using NumPy \cite{VanDerWalt2011}, and scaffold representations drawn using NetworkX \cite{Hagberg2008}. It computes all the results on information presented in the current study, including the information paths, statistical tests of $I_k$ values  described in the next sections and the finite entropy rate $\frac{H_k}{k}$. The input is an excel table containing the data values, e.g. the matrix $D$ with the first row and column containing the labels. Here, we limited our analysis to $n=21$ genes of specific biological interest.\\

\subsection{Estimation of the undersampling dimension: statistical result} \label{Undersampling dimension_gen}

The information data analysis presented here depends on the two parameters $N$ and $m$. The finite size of the sample $m$ is known to impose an important bias in the estimation of information quantities: in high-dimensional data analysis, it is quoted as the Hugues phenomenon \cite{Hughes1968} and in entropy estimation it has been called the sampling problem since the seminal work of Strong and colleagues \cite{Strong1998,Nemenman2004,Merchan2016}. For the method we suggested, it is important to notice that the size $m$ of the population $Z$ is in general much smaller than the dimension of the probabilty simplex $N=N_1...N_n-1$. For instance, in the mentioned study of genes as variables \cite{Tapia2018}, we had $m=111$ for $DA$ neurons (resp. $m'=37$ for $NDA$ neurons) as respective number of neurons, but $N=9^{21}-1$, because we could only achieve the computation for the 21 most relevant genes. In the example considering cells as variables presented here in Figure \ref{figure_Supp_Result_dopa_nondopa_infopath}, the situation is even worse, with a sample size of $m=41$ genes and a dimension of $N=9^{20}-1$ as only 20 cells were considered. Thus the pure entropies $H_k,k=1,...,n$ must satisfy the following inequality:
\begin{equation}
\forall J\subset [n], k=| J|=card J,\quad H_k(X_J;P)\leq \log_2 m.
\end{equation}
where equality is an extreme signature of undersampling.
However, suppose that all the numbers $N_i,i=1,...,n$ are equal to $r\geq 2$, the maximum value of $H_k$ is equal to $k\log_2 r$, for instance $2k.\log_2(3)$ in the example.\\

\begin{lemma}\label{curse}
	Take the uniform probability on the simplex $\Delta([n])$ with affine coordinates, and take $\epsilon$ such that $0< \epsilon\leq 1/e\approx 0,367$; then the probability that $H_k(X_J)$
	is greater than $\epsilon k\log_2 r$ is larger than $1-\epsilon$.
\end{lemma}

\begin{proof} 
	Concerning $H_k$, the simplex $\Delta([n])$ is replaced by $\Delta([k])$; then consider the set $\Delta_\epsilon$ of probabilities such that $p_j\geq \epsilon r^{-k}$ for any coordinate $j$ between $1$ and $r^{k}$, this set is the complement of the union of the sets $X_j(\varepsilon),i=1,...,r^{k}$ where $p_j< \epsilon r^{-k}$. From the properties of volumes in affine geometry, the measure of each set $X_j(\varepsilon)$ is less than $\epsilon r^{-k}$, thus the probability of $\Delta_\epsilon$ is larger than $1-\epsilon$. And for any index $j$ the monotony of $-x\ln x$ between $0$ and $1/e$ implies
	\begin{equation}
	-p_j\log_2 p_j>  \epsilon r^{-k} k\log_2 r;
	\end{equation}
	then by summation over all the indices we obtain the result.
\end{proof}

\noindent By example, for $r=9$, and $\epsilon=1/e$, this gives that $H_k\geq 2k\log_2(3)/e$ is two times more probable than the opposite.\\

Consequently, in the above experiment, the quantities $H_k$, then $I_k$ are not significant, except if they appear to be significantly smaller than $\log_2 m$.\\

\noindent In counterpart, as soon as the measured $H_k$ is inferior than the predicted one for $m$ values, this is significant. Note that the lemma \ref{curse}, with $n$ replaced by $m$, gives estimations for the entropies of raw data. In the next section, we propose a computational method to estimate the dimension $k_u$ above which information estimation ceases to be significant.

\subsection{Estimation of the undersampling dimension: Computational result}\label{Undersampling dimension}

Following the original presentation of the sampling problem by Strong and colleagues \cite{Strong1998}, the extreme cases of sampling are given by:
\begin{itemize}
	\item When $N_1=N_2=...=N_n=1$, there is a single box $\Omega$ and $P(\Omega)=m/m=1$ and we have $H_k=I_k=0, \forall k \in {0,...,n}$. The case where $m=1$ is identical. This fixes the lower bound of our analysis in order not to be trivial; we need $m\geq2$ and $N_1=N_2=...=N_n\geq2$. 
	\item When $N_1.N_2...N_n$ are such that only one data point falls into a box, $m$ of the values of atomic probabilities are  $1/m$ and $N_1.N_2...N_n-m$ are null as a consequence of equation \ref{marginalization}, and hence we have $H_n=\log_2 m$.
\end{itemize}	
Whenever this happens for a given k-tuple, all the $HP_k$  paths passing by this k-tuple will stay on the same information values since conditional entropy is non-negative: we have $H_k=H_{k+1}$ or equivalently $(X_1,...,X_k)H(X_{k+1})=0$, and all $k+l$-tuples are deterministic (a function of) with respect to the k-tuple. This is typically the case illustrated in Figure \ref{figure_Supp_Result_dopa_nondopa_infopath}: adding a new variable to an undersampled k-tuple is equivalent to adding the deterministic variable "0" since the probability remains unchanged ($1/m$). 
\begin{figure} [!h]
	\centering
	\includegraphics[height=10cm]{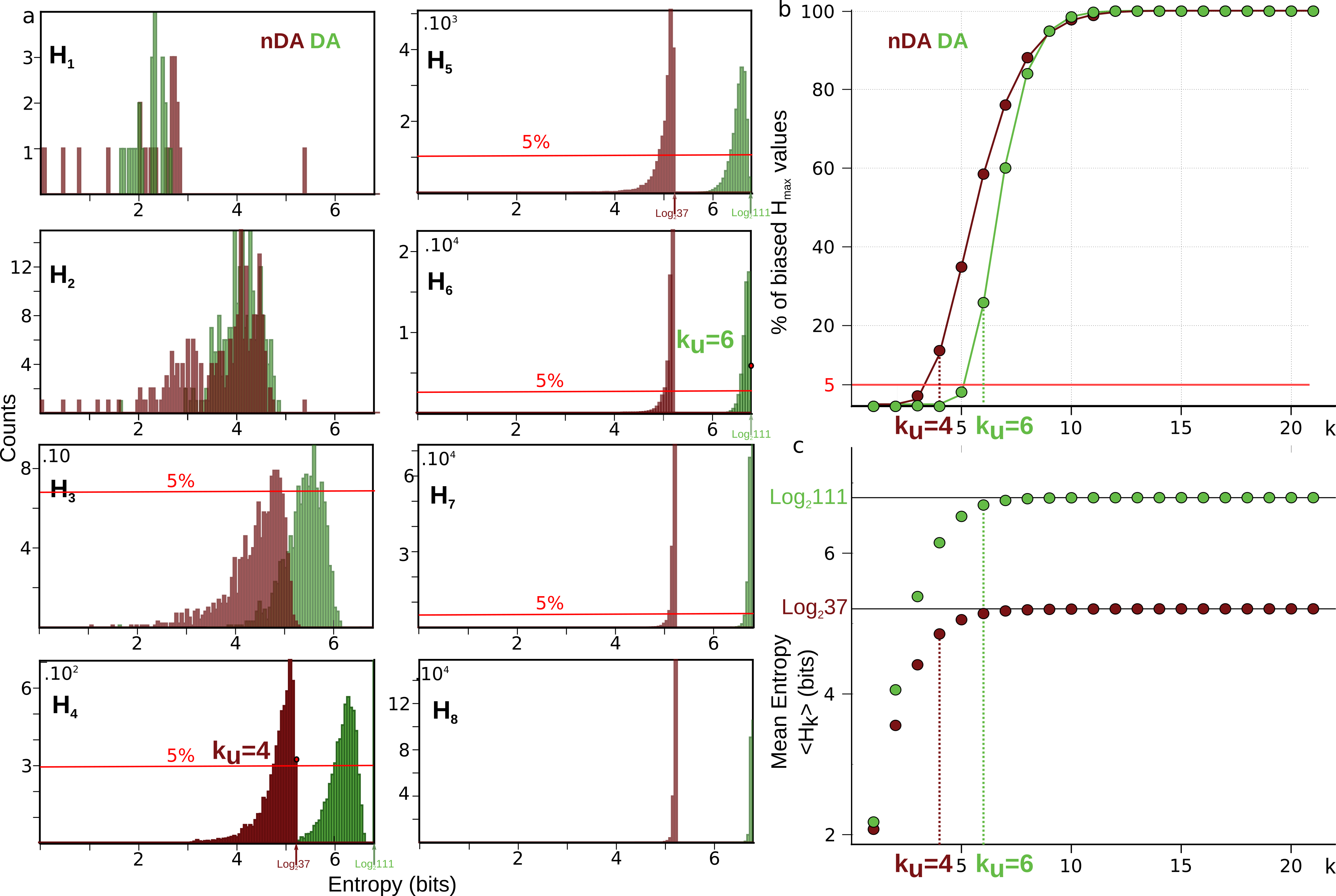}
	\caption{\textbf{Determination of undersampling dimension $k_u$. a,} distributions of $H_k$ for $m=111$ population A neurons (green) and $m=37$ population B neurons (dark red) for $k=1,..,6$. The horizontal red line represents the threshold we have fixed to 5 percent of the total number of k-tuples. \textbf{c,} Plot of the percent of maximum entropy $H_k=\ln m$ biased values as a function of the dimension $k$. The horizontal red line represents the threshold fixed to 5 percent, giving $k_u=6$ for population A and $k_u=4$ for population B  neurons. \textbf{c,} The mean $\langle HP \rangle(k)$ paths for these two populations of neurons, the maximum entropy $H_k=\ln m$ is represented by plain horizontal lines.}
	\label{figure_Supp_undersampling}
\end{figure}

Considering the analysis of cells as variables (matrix $D^T$), the signature of this undersampling is the saturation at   $H_k=\log_2 41$ observed in the $H_k$ landscape in Figure \ref{figure_Supp_Result_dopa_nondopa_infopath}b, starting at $k=5$ for some 5-tuples of neurons. Considering the analysis of genes as variables (matrix $D$ \cite{Tapia2018}), the mean entropy computed also shows this saturation at  $H_k=\log_2 111$ for population A neurons  and $H_k=\log_2 37$ for population B neurons. 
We propose to define a dimension $k_u$ as the dimension for which the probability $p_u$ of having the $H_k$ at the biased value of $H_k=\log_2 m$ is above $5$ percent ($p_u=0.05$). As shown for the analysis of cells as variables  in Figure \ref{figure_Supp_undersampling}, this basic estimation gives here $k_u=6$ for population A neurons and $k_u=4$ for population B neurons. The information structures identified by our methods beyond these values can be considered as unlikely to have a biological or physical meaning and shall not be interpreted. Since undersampling mainly affects the distribution of $I_k$ values close to 0 value, the maxima and minima of $I_k$ and the maximal and minimal information paths below $k_u$ are the least affected by the sampling problem and the low sample size. This will be illustrated in the next sections.\\

\subsection{k-dependence test} \label{k-independence test}

Pethel and Hahs \cite{Pethel2014} have constructed an exact test of 2-dependence for any pair of variables, not necessarily binary or iid. Indeed, the iid condition usually assumed for the $\chi^2$ test does not seem relevant for biological observations and the examples given here  and in \cite{TapiaPacheco2017,Tapia2018} with genetic expression support such a general statement. It allows to test the significance of the estimated $I_2$ values given a finite sample size $m$, the null hypothesis being that $I_2=0$  (2-independence according to Pethel and Hahs).  We follow here their presentation of the problem, and provide an extension of their test to arbitrary $k$ (higher dimensions), with the null hypothesis being the k-independence $I_k=0$.
Even in the lowest dimensions, and below the undersampling bound, the values of $I_k$ estimated from a finite sample size $m$ are considered as biased \cite{Pethel2014}. If one considers an infinite sample ($m\rightarrow \infty$) of n independent variables, we then have for all $k\geq2$ $I_k=0$. However, if we randomly shuffle the values such that the marginal distributions for each variable $X_i$ are preserved, the estimated $I_k$ can be very different from 0, with distributions of $I_k$ values not centered on 0. Figure \ref{figure_independencetest} illustrates an example of such bias with $m=111$ for the analysis with genes as variables.\\

\begin{figure} [!h]
	\centering
	\includegraphics[height=8.7cm]{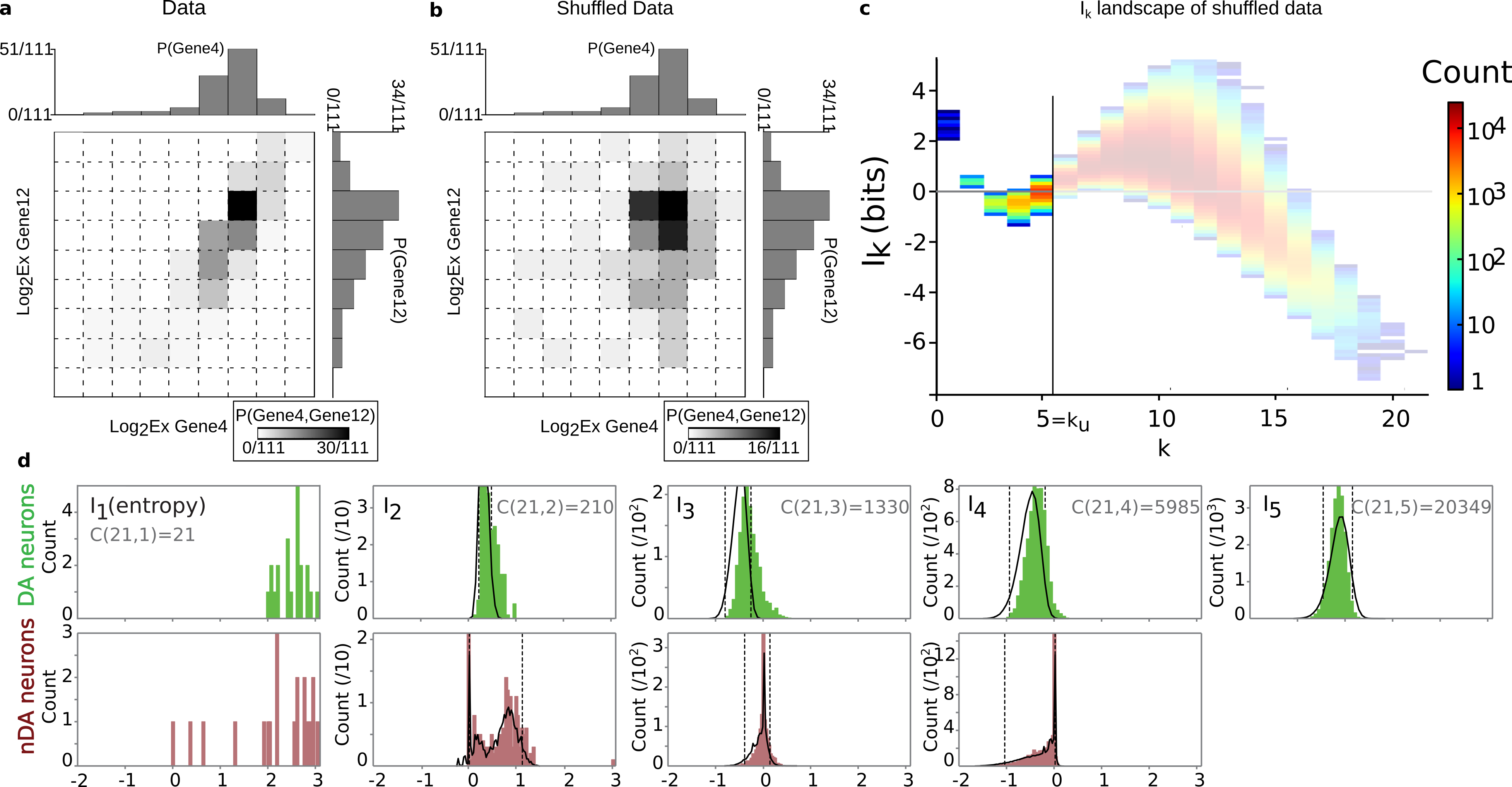}
	\caption{\textbf{Probability and Information landscape of shuffled data.} The figure corresponds to the case of analysis with genes as variables.  \textbf{a,} joint and marginal distributions of two genes (genes 4 and 12) for $m=111$ population A neurons. \textbf{b,} joint and marginal distributions after a shuffling of the values of expression of each gene. \textbf{c,} the estimated $I_k$ landscape for the expression of 21 genes after shuffling. \textbf{d,} histograms representing the distribution of $I_k$ values for all the degrees until $k=5$ for population B. The total number of combinations C(n,k) for each degree (number of pairs for $I_2$; number of triplets for $I_3$, etc) is given in gray. The averaged shuffled values of information obtained with 17 shuffles are represented on each histogram as a black line, and the statistical significance threshold values for $p=0.1$  are represented as vertical dotted lines.}
	\label{figure_independencetest}
\end{figure}
Reproducing the method of Pethel and Hahs \cite{Pethel2014}, we designed a shuffling procedure of the $n$ variables, which consists in randomly permuting the measured values (co-ordinates) of each variable one by one in the matrix $D$ or $D^T$ (geometrically, a "random" permutation of the co-ordinates of each data point, point by point). Such a shuffle leaves marginal probabilities invariant. Figure \ref{figure_independencetest} gives an example of the joint and marginal distributions before and after shuffle for two genes. Extending the 2-test of \cite{Pethel2014} to $k\geq2$, the $I_k$ values obtained after shuffling provide the distribution of the null hypothesis, k-independence ($I_k=0$) according to \cite{Pethel2014}. The task is hence to compute many shuffles, 10.000 in \cite{Pethel2014}, in order to obtain these "null" distributions.
The exact procedure of Pethel and Hahs \cite{Pethel2014} would require to obtain such "null" distribution for all the $2^n$ tuples, which would require a number of shuffled trials impossible to obtain computationally. We hence propose a global test that consists in computing 17 different shuffles of the 21 genes, giving "null" distribution of shuffled $I_k$ values composed of $21\times\binom{n}{k}$. For example, the test of $2$-dependence and $3$-dependence will be against a null distribution with $21*210=3750$ $I_2$ values and $21*1330=22610$ $I_3$ values respectively. We fix a p value above which we reject the null hypothesis (a significance level, fixed at $p= 0.05$ in \cite{Pethel2014}), allowing to determine the statistical significance thresholds as information values for which the integral of the null distribution reaches the significance level $p= 0.05$. This holds for $k=2$, as described in \cite{Pethel2014}, but since for $k\geq 2$ $I_k$ can be negative, the test becomes symmetric on the distribution, and hence for $k\geq 2$ we choose a significance level of $p= 0.1$ in order to stay consistent with the 2-dependence test. The "null" distributions and the threshold  given by the significance p-value of rejection are illustrated in Figure \ref{figure_independencetest}d.
If the observed values of $I_k$ are above or below these threshold values, we reject the null hypothesis.\\
In practice, a random generator is used to generate the random permutations (here the NumPy generator \cite{VanDerWalt2011}), and the present method is not exempt from the possibility that it generates statistical dependences in the higher degrees.\\
\textbf{Interpretation of the dependence test}.
The original interpretation of the test by Pethel and Hahs was that the null hypothesis corresponded to independent distributions, motivated by the statement that "permutation destroys any dependence that may have existed between the datasets but preserves symbol frequencies". However, considering simple analytical examples could not allow us to confirm their statement. We propose that for a given finite $m$, random permutations express all the possible statistical dependences  that preserve symbol frequencies (cf. the discussion of E.Borel in \cite{Borel1913}). This statement basically corresponds to what we observe in Figure \ref{figure_independencetest}. Hence we propose that in finite context the null-hypothesis corresponds to a random k-dependence. The meaning of the presented test is hence a selectivity or specificity test: a test of an $I_k$ of given k-tuple against a null hypothesis of "randomly" selected k-statistical dependences that preserve the marginals and $m$.   \\

\subsubsection{Sampling size and graining landscapes - stability of minimum energy complex estimation} \label{Sampling size}

Figure \ref{Figure_m_N_infopath} gives a first simple study of how robust the paths of maximum length are with respect to the variations of $m$ and $N$, in the case of the analysis of genes as variables. The limit $N\rightarrow \infty$ recovers Riemann integration theory and gives the differential entropy with the correcting additive factor $N$ (theorem 8.3.1 \cite{Cover1991}).
\begin{figure} [!h]
	\centering
	\includegraphics[height=12.5cm]{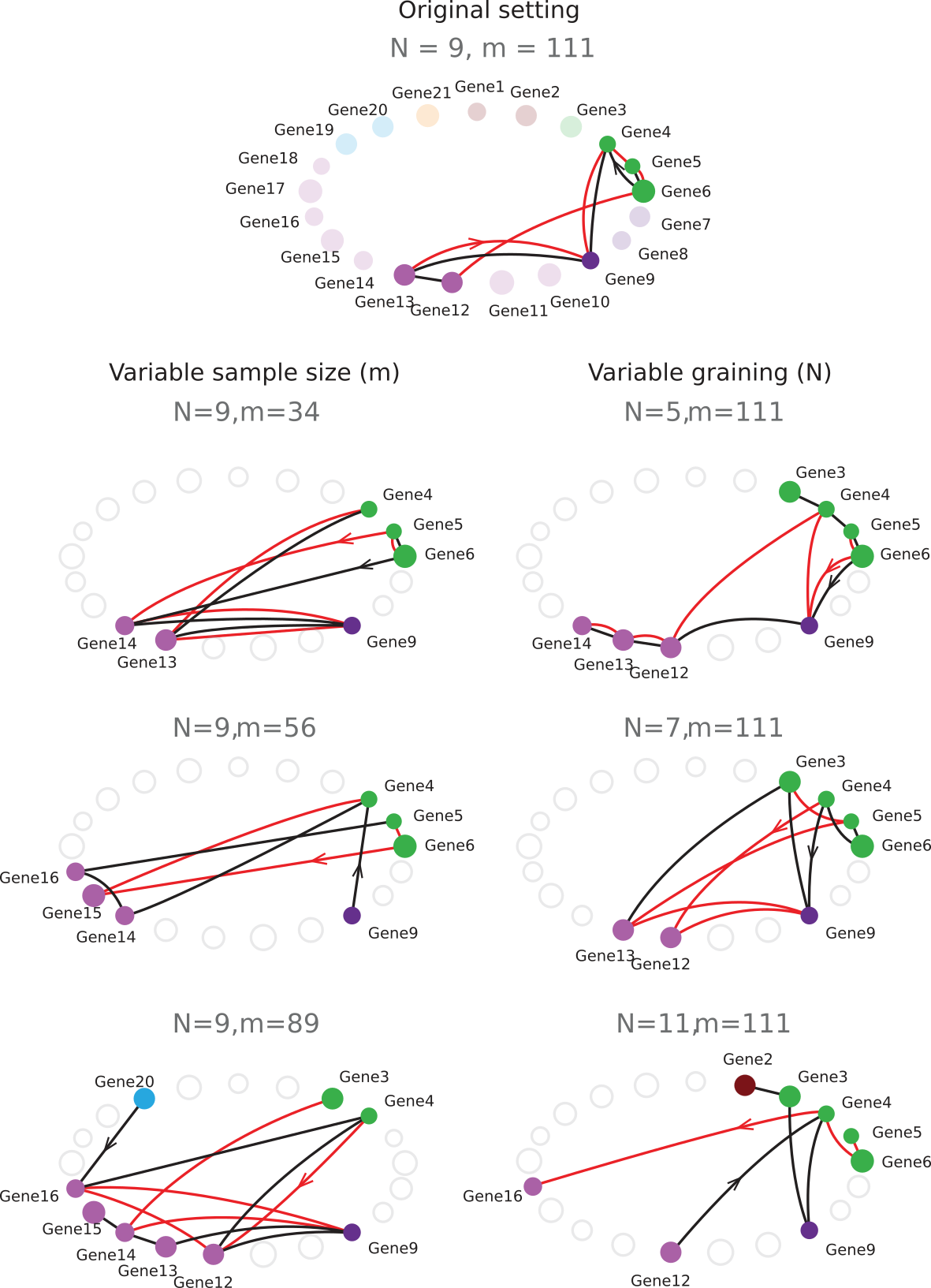}
	\caption{\textbf{Effect of changing sample size and graining on the identification of gene modules}. The figure corresponds to the case of analysis with genes as variables for the population A neurons. The positive $I_k$ paths of maximum length were computed for a variable number of cells ($m$,left column) and a variable graining ($N$, right column). For clarity, only the two positive paths of maximum length are represented (first in red, second in black) for each parameter setting and the direction of each path is indicated by arrowheads. The two positive paths of maximum length for the original setting ($N=9$, $m=111$) are represented on the scaffold at the top of the figure for comparison. Smaller samples of cells (one random pick of $34$, $56$ and $89$ cells) and larger ($N=11$) or smaller ($N=5,N=7$) graining than the original ($N=9$) were tested. Although slight differences in paths can be seen (especially for $N=11$), most of the parameter combinations identify gene modules that strongly overlap with the module identified using the original setting.}
	\label{Figure_m_N_infopath}
\end{figure}

The information paths of maximal length identified by our algorithm are relatively stable in the range of $N=5,7,9,11$ and $m=34,56,89,111$ where the $m$ cells were taken among the 111 neurons of population A. If we consider that the paths that only differ by the ordering of the variables are equivalent, then the stability of the two first paths is further and largely improved.
The undersampling dimension obtained in these conditions is $k_u(m=34)=5,~ k_u(m=56)=6,~ k_u(m=89)=6,~k_u(m=111)=6$ and $k_u(N=5)=8,~ k_u(N=7)=7,~ k_u(N=9)=6,~ k_u(N=11)=5$.
In general, information landscapes can be investigated with the additional dimensions of $N$ and $m$ together with $n$. It allows to define our landscapes as iso-graining landscapes and to study the appearance of critical points in a way similar to what is done in thermodynamics. In practice, to study more precisely the variations of information depending on $N$ and $m$ and to obtain a 2-dimensional representation, we plot the mean information as a function of $N$ and $m$ together with $n$, as presented in Figure \ref{figure_Supp_graining}a. We call the obtained landscapes the iso-graining $I_k$ landscapes. The choice of a specific graining $N$ can be done using this representation: a "pertinent" graining should be at a critical point of the landscape (\textbf{a} first minimum of an information path), consistent with the proposition of the work of Reshef and colleagues \cite{Reshef2011}, who used maximal information coefficient ($MI_2C$) depending on the graining (with a more elaborated graining procedure) to detect pairwise associations. We have chosen to illustrate the landscapes with $N=9$ according to this criterion and the undersampling criterion, because the $I_2$ values are close to their maximal values and the sampling size is not too limiting, with a $k_u=6$ (see Figure \ref{figure_Supp_graining}\textbf{a}). Moreover, this choice of graining size $N=9$ is sufficiently far from the critical point to ensure that we are in the condensed phase where interactions are expected. It is well below the analog of the critical temperature (the critical graining size), which according to the Figure \ref{figure_Supp_graining}a happens at $N_c=3$ (the $N$ for which the critical points cease to be trivial). In general, there is no reason why there should be only one "pertinent" graining.   
\begin{figure} [!h]
	\centering
	\includegraphics[height=6cm]{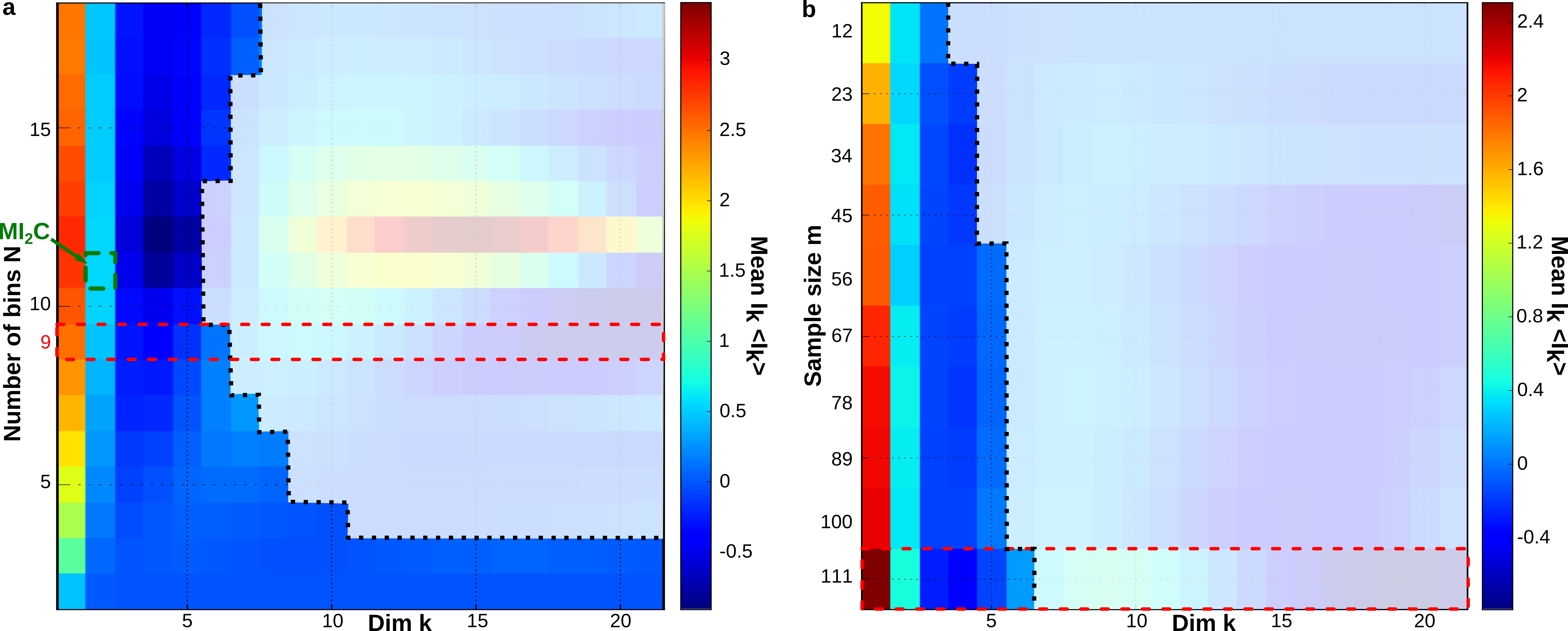}
	\caption{\textbf{Iso-sample-size ($m$) and iso-graining mean $\langle IP \rangle(k)$ landscapes.} The figure corresponds to the case of analysis with genes as variables for the population A neurons. \textbf{a,} The mean $\langle IP \rangle(k)$ paths are presented for $N=2,...,18$ and $n=21$ genes for the $m=111$ population A neurons. The "undersampling" region beyond the $k_u$ is shaded in white and delimited by black dotted line (the $k_u$ was undetermined for  $N=2,3$). For $N=2$ the mean $\langle IP \rangle(k)$ path has no non-trivial minimum (monotonically decreasing). This $N=2$ iso-graining is analog to the non condensed disordered phase of non interacting bodies, $ \forall k>1, ~ \langle IP \rangle(k)\approx 0$. All the other mean $\langle IP \rangle(k)$ paths have non-trivial critical dimensions. The condition $N=9$, $m=111$ used for the analysis is surrounded by dotted red lines. It was chosen to be in the condensed phase above the critical graining, here $N_c=3$, close to the criterion of maximal mutual information coefficient $MI_2C$ proposed by Reshef and colleagues  (bin surrounded by green dotted line) and with a not too low undersampling dimension.  \textbf{b,} The mean $\langle IP \rangle(k)$ paths are presented for $m=111,100,...,12$  population A neurons and $n=21$ genes with a number of bins $N=9$.}
	\label{figure_Supp_graining}
\end{figure}
The graining algorithm could be improved by applying direct methods of probability density estimation \cite{Scott1992}, or more promisingly persistent homology \cite{Epstein2011}. Finer methods of estimation (graining) have been developed by Reshef and colleagues \cite{Reshef2011} in order to estimate pairwise mutual-information, with interesting results. Their algorithm  presents a lower computational complexity than the estimation on the lattice of partitions, but a higher complexity than the simple one applied here. \\
What we call the iso-sampling size $I_k$ landscapes is presented in Figure \ref{figure_Supp_graining}b for mean $I_k$. Such investigation is also important since it monitors what is usually considered as the convergence (or divergence) in probability of the informations. For the estimations below the $k_u$ represented here, the information estimations are quite constant as a function of $m$, indicating the stability of the estimation with respect to the sample size.

\section*{Acknowledgement}
\textbf{Acknowledgments:} This work was funded by the European Research Council (ERC consolidator grant 616827 \textit{CanaloHmics} to J.M.G.; supporting M.T. and P.B.) and Median Technologies, developed at Median Technologies and UNIS Inserm 1072 - Universit\'{e} Aix-Marseille, and at Institut de Mathématiques de Jussieu - Paris Rive Gauche (IMJ-PRG), and thanks previously to supports and hostings since 2007 of Max Planck Institute for Mathematic in the Sciences (MPI-MIS) and Complex System Instititute Paris-Ile-de-France (ISC-PIF). This work addresses a deep and warm acknowledgement to the researchers who  helped its realization: G.Marrelec and J.P. Vigneaux; or  supported-encouraged it: H.Atlan, F.Barbaresco, H.B\'{e}nali, P.Bourgine, F.Chavane, J.Jost, A.Mohammad-Djafari, JP.Nadal, J.Petitot, A.Sarti, J.Touboul. \\

\section*{Supplementary material, contributions and previous versions}
\textbf{Previous Version:} a partial version of this work has been deposited in the method section of Bioarxiv 168740 in July 2017 and preprints \cite{Baudot2018}.\\
\textbf{Author contributions:} D.B. and P.B. wrote the paper, P.B. analysed the data ; M.T. performed the experiments ; M.T. and J.M.G. conceived and designed the experiments ; D.B., P.B., M.T. and J.M.G participated in the conception of the analysis. \\
\textbf{Supplementary material:} The software Infotopo is available at https://github.com/pierrebaudot/INFOTOPO\\

\section*{Abbreviations}
The following abbreviations are used in this manuscript:\\

\noindent 
\begin{tabular}{@{}ll}
iid & independent identicaly distributed\\
DA & Dopaminergic neuron\\
nDA & non Dopaminergic neuron\\
$H_k$ & Multivariate k-joint Entropy\\
$I_k$ & Multivariate k-Mutual-Information\\
$G_k$ & Multivariate k-total-correlation or k-multi-information\\
$MI_2C$ & Maximal 2-mutual-Information Coefficient 
\end{tabular}

\appendix

\section{Appendix: Bayes free energy and Information quantities} \label{Appendix: Bayes free energy}

\subsection{Parametric modelling}

As we mentioned in the introduction, the statistical analysis of data $X$ is confronted to a serious the risk of circularity, because the confidence in the model is dependent on the probability law it assumes and reconstructs in part. Several approaches were followed to escape from this circularity; all of them rely on the choice of a set $\Theta$ of probability laws where $\mathbb{P}_X$ is researched. For instance, maintaining the frequentist point of view, the Fisher information metric on $\Theta$ (cf. \cite{Ly2017}) determines bounds on the confidence. Another popular approach is to choose an \emph{a priori} probability $\mathbb{P}_\Theta$ on $\Theta$, and to revise this choice after all the experiments $X(z),z\in Z$, by computing the probability on $E\times\Theta$, which better explains the results (the new probability on $\Theta$ is its marginal, and for each $\theta$ in $\Theta$, the probability $P_\theta$ on $E$ is its conditional probability).
Here a more precise principle is necessary, which expresses a trade-off between the maximization of the marginal probability of the results under the constraint to be not too far from the prior. A popular example is the minimization of the Bayes Free energy $F_V(P)$, which appears as the maximum of entropy of the new \emph{a posteriori} probability under the constraint to predict in the mean the data and to depart the less possible from the \emph{a priori} probability on the probabilities. This function is given by a Kullback-Leibler distance $D_{KL}$. In the finite setting, with a uniform a priori, this consists in maximizing the entropy among the laws that predict the observed distribution. Remark that the two methods, Bayes and Fisher, are related, because in most cases
the chosen \emph{a priori} probability laws (and the data estimation) used in the function $F_V$ are given by frequencies, and
because the distance $D_{KL}(P,Q)$ is approximated by the Fisher metric at $P$ when $Q$ approaches $P$. \\

\subsection{Bethe approximation}

Let us remind that for two probability laws $P,Q$ on the same finite set $\Omega$, the Kullback-Leibler divergence from $P$ to $Q$ is defined by
\begin{equation}
D_{KL}(P,Q)=\sum_{x\in\Omega} P_x \ln\frac{P_x}{Q_x}=\mathbb{E}_P(-\ln Q)-H(P).
\end{equation}
Contrarily to its name, it is not a true distance, because it is not symmetric, however it is always positive
and it is equal to zero if and only if $P=Q$. Another drawback is that it can be $+\infty$: it is so when there exists $x$
such that $Q_x=0$ but $P_x> 0$, i.e. when $P$ is not absolutely continuous with respect to $Q$.\\
\indent The Kullback-Leibler divergence permits to define the Bayes free energy functional as follows:\\

\noindent The unknown is the probability law $P_b$ on $E\times\Theta$.
\begin{equation}
F_V(P_b)=D_{KL}(P_b,P_L\otimes P_a)=\sum_{x_L,\theta}(\ln\frac{P_b(x_L,\theta)}{P_L(x_L)P_a(\theta)})P_b(x_L,\theta),
\end{equation}
where $P_a(\theta)$ is the \emph{a priori} on the probability laws and where $P_L(s)$ represents the new partial data, collected by
a collection of variables $X_L$, and expressed by a probability law.\\
\begin{equation}
F_V(P_b)=\mathbb{E}_{P_b}(-\ln P_a+D_{KL}((X_L)*P_\theta,P_L))-H(P_b).
\end{equation}
This function looks like a free energy in Statistical Physics, that is the sum of the negentropy and
the mean of an energy function.\\
Here we assume that $\Omega=E_S$ for a family of variables
$S_i,i=1,...,N$, and the states are the possible values of the joint variable $S$.\\
\noindent Due the strict convexity of the negentropy, $F_V$ has a unique minimum, that defines the equilibrium state.\\
\indent Practically, the full entropy is difficult to estimate, thus approximations were introduced, following Bethe and Kikuchi (cf. Mori \cite{Mori2013}), generalizing the Mean Field
Theory. These approximations are no more convex in the unknown $P_b$, they are obtained by replacing the full entropy $H$ by a convenient linear combination of entropies of more accessible variables (observable quantities). It is here that the information functions $H_k$ and $I_k$ appear in the bayesian variational
calculus (cf. Mori \cite{Mori2013}):\\

\noindent Consider a simplicial complex $K$ in the simplex $\Delta([N])$, i.e. a collection of faces that contains every faces inside each face it contains, and assume $K$ a combinatorial ($PL$)  manifold of  dimension $d$, with possibly a boundary that is a combinatorial ($PL$)  manifold $\partial K$; then the Bethe function associated to $K$ is given by the two equivalent following formulas:
\begin{equation}
F_B(Q)=\mathbb{E}_Q(-\ln f)-\sum_{I \in K^{*}}(-1)^{d-|I|}H(S_I),
\end{equation}
where the sum is taken over the set $K^{*}$ of faces not contained in $\partial K$, and $|I|$ denotes the dimension of the face $I$;
\begin{equation}\label{bethelisseinfos}
F_B(Q)=\mathbb{E}_Q(-\ln f)-\sum_{J \in K}(-1)^{|J|+1}I_{|J|}(S_I;Q),
\end{equation}
where the sum is taken over all the faces of $K$, including the boundary, and $I_{|J|}(S_J;Q)$ is the higher mutual information considered everywhere above in the text.\\

\bibliographystyle{acm}
\bibliography{bibtopo}

\end{document}